\documentclass[aip,reprint]{revtex4-2}

\usepackage{graphicx} 

\usepackage[dvipnames]{xcolor}

\usepackage{algorithm}
\usepackage{algorithmicx}
\usepackage[noend]{algpseudocode}

\usepackage{amsmath,amsthm,amssymb}

\newtheorem{theorem}{Theorem}

\usepackage{newtxtext,newtxmath}

\newcommand{\Emin}{{E}_{\textup{min}}}
\newcommand{\Emax}{{E}_{\textup{max}}}
\newcommand{\Eminh}{a}
\newcommand{\Emaxh}{b}

\newcommand{\KPM}{\textup{KPM}}

\usepackage{bm}
\renewcommand{\vec}{\bm}

\usepackage{hyperref}
\usepackage[capitalize,nameinlink,noabbrev]{cleveref}
\crefformat{equation}{(#2#1#3)}
\crefmultiformat{equation}{(#2#1#3)}%
{ and~(#2#1#3)}{, (#2#1#3)}{ and~(#2#1#3)}

\renewcommand{\d}{\mathrm{d}}

\makeatletter
\def\@email#1#2{%
 \endgroup
 \patchcmd{\titleblock@produce}
  {\frontmatter@RRAPformat}
  {\frontmatter@RRAPformat{\produce@RRAP{*#1\href{mailto:#2}{#2}}}\frontmatter@RRAPformat}
  {}{}
}%
\makeatother

\usepackage[normalem]{ulem}

\begin{document}
\title{A spectrum adaptive kernel polynomial method}

\author{Tyler Chen}
\email{tyler.chen@nyu.edu}
\homepage{https://research.chen.pw}
\affiliation{New York University}
\altaffiliation{Department of Mathematics, Courant Institute of Mathematical Sciences, New York University, 251 Mercer Street, New York, NY 10012}%
\altaffiliation{Department of Computer Science and Engineering, Tandon School of Engineering, New York University, 370 Jay Street, New York, NY 11201}%

\date{\today}

\begin{abstract}
The kernel polynomial method (KPM) is a powerful numerical method for approximating spectral densities. Typical implementations of the KPM require an a prior estimate for an interval containing the support of the target spectral density, and while such estimates can be obtained by classical techniques, this incurs addition computational costs. We propose an spectrum adaptive KPM based on the Lanczos algorithm without reorthogonalization which allows the selection of KPM parameters to be deferred to after the expensive computation is finished. Theoretical results from numerical analysis are given to justify the suitability of the Lanczos algorithm for our approach, even in finite precision arithmetic. While conceptually simple, the paradigm of decoupling computation from approximation has a number of practical and pedagogical benefits which we highlight with numerical examples.
\end{abstract}

\maketitle

\section{Introduction}

Over the past several decades, a number of moment-based schemes have been developed to approximate spectral densities of matrices.\cite{skilling_89,jaklic_prelovsek_94,bai_fahey_golub_96,lin_saad_yang_16,jin_willsch_willsch_lagemann_michielsen_deraedt_21}
These methods typically access the matrix of interest through matrix-vector products and are therefore especially well suited for applications where matrix-vector products can be performed quickly or in which the matrix of interest is so large that exact diagonalization techniques are infeasible.
Perhaps the most prominent of these methods is the Kernel Polynomial Method (KPM),\cite{skilling_89,silver_roder_94,silver_roeder_voter_kress_96,weisse_wellein_alvermann_fehske_06,lin_saad_yang_16} which has found widespread use in quantum physics/chemistry \cite{skilling_89,silver_roder_94,silver_roeder_voter_kress_96,weisse_wellein_alvermann_fehske_06,ganeshan_pixley_dassarma_15,garca_covaci_rappoport_15,lin_saad_yang_16,carr_massatt_fang_cazeaux_luskin_kaxiras_17,carvalho_garciamartinez_lado_fernandezrossier_18,varjas_fruchart_akhmerov_perezpiskunow_20} and beyond. \cite{han_malioutov_avron_shin_17,dong_benson_bindel_19}

A standard implementation of the KPM\cite{weisse_wellein_alvermann_fehske_06} produces an approximate spectral density via a truncated Chebyshev polynomial expansion obtained using the low-degree Chebyshev moments of the target spectral density.
These moments are computed using a Chebyshev recurrence shifted and scaled to an interval of approximation containing the support of the target spectral density.
If this interval does not contain the support of the target spectral density, the KPM approximation is unlikely to converge, but if the interval is too large, then the KPM approximation will lose resolution.
All implementations of the KPM that we are aware of require this interval of approximation to be determined ahead of time.
Thus, as a pre-processing step, it is typical to run another algorithm to determine bounds for the support of the target spectral density.

The focus of this paper is an an implementation of the KPM which decouples computation from the choice of approximation method.
For instance, our implementation allows many different intervals of approximation to be tested out at essentially no cost once the main computation is completed.
In fact, and more importantly, approximations corresponding to different families of polynomials can also be efficiently obtained.
The choice of polynomial family can significantly impact the qualitative properties of the resulting KPM approximation, but the use of non-standard orthogonal polynomial families has been limited in practice thus far, arguably due to the previous lack of a simple implementation.

Our spectrum adaptive KPM is based on the Lanczos algorithm \emph{without reorthgonalization}.
It is well-known that the Lanczos algorithm is unstable, and this has lead to a general hesitance to use Lanczos-based methods for approximate spectral densities unless reorthogonalization is used.\cite{jaklic_prelovsek_94,silver_roeder_voter_kress_96,aichhorn_daghofer_evertz_vondelinden_03,weisse_wellein_alvermann_fehske_06,ubaru_chen_saad_17,granziol_wan_garipov_19}
Amazingly, this instability does not limit the usefulness of the Lanczos algorithm for many tasks. 
We use theoretical results from the numerical analysis literature to justify the validity of our approach in finite precision arithmetic.
We believe that this commentary will be of general value to the computational physics/chemistry communities.
Complimentary numerical experiments provide empirical evidence of the stability of the proposed algorithm.

Finally, we note that there are a number of related Lanczos-based algorithms such as the Finite Temperature Lanczos Method (FTLM);\cite{jaklic_prelovsek_94} see also Stochastic Lanczos Quadrature. \cite{bai_fahey_golub_96,ubaru_chen_saad_17}
These methods are also widely used in practice, and are viewed by some as preferable to KPM based methods in many settings.\cite{schnack_richter_steinigeweg_20,morita_tohyama_20,chen_trogdon_ubaru_22}
We do not advocate for the use of the KPM over these methods nor for the use of any of these methods over KPM.
Rather, our aim is to provide a new tool which allows practitioners to test out all of these algorithms for essentially free.
For instance, users no longer need to make an a priori decision to use FTLM or KPM; they can simply output both approximations and then decide which to use later.

\begin{figure*}[tb]
    \includegraphics[scale=.6]{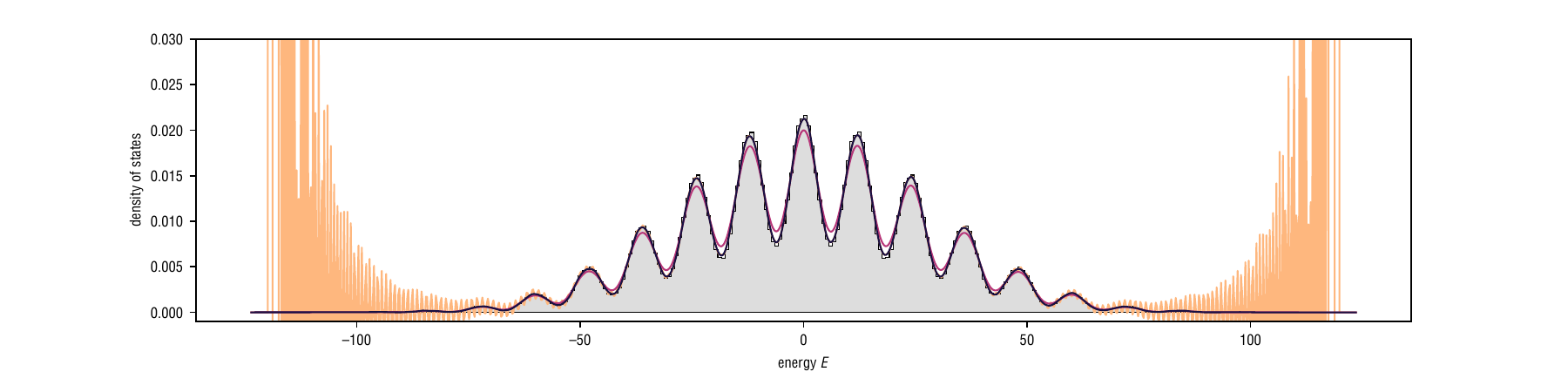}
    \caption{Illustration showing the impact of the support $[\Eminh,\Emaxh]$ for the standard Chebyshev KPM with $s=500$ moments and Jackson's damping kernel.
    Here $[\Eminh,\Emaxh] = [\Emin-\eta,\Emax+\eta]$ for varying choices of $\eta$.
    \emph{Legend}: 
    $\eta = 0$
    ({\protect\raisebox{0mm}{\protect\includegraphics[scale=.7]{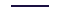}}}), 
    $\eta = -0.00007(\Emax-\Emin)$
    ({\protect\raisebox{0mm}{\protect\includegraphics[scale=.7]{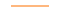}}}),
    $\eta = 0.5(\Emax-\Emin)$
    ({\protect\raisebox{0mm}{\protect\includegraphics[scale=.7]{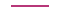}}}),
    histogram of true eigenvalues
    ({\protect\raisebox{0mm}{\protect\includegraphics[scale=.7]{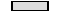}}}).
    \emph{Takeaway}:
    Observe that even a slight underestimate of $[\Emin,\Emax]$ results in a loss of convergence while an overestimate of $[\Emin,\Emax]$ results in a loss of resolution.
    Our spectrum adaptive KPM allows the reference density $\sigma(E)$ to be chosen after computation. 
    In fact, many different $\sigma(E)$ can be efficiently obtained and compared.
    }
    \label{fig:support}
\end{figure*}

\section{The kernel polynomial method}

We present the KPM from the perspective of orthogonal polynomials.
Throughout, $\vec{H} = \sum_{n=1}^{d} E_n |\vec{u}_i \rangle \langle \vec{u}_i|$ will be a Hermitian Hamiltonian of finite dimension $d<\infty$ with corresponding Density of States (DOS)
\begin{equation}
\rho(E) = \frac{1}{d} \sum_{n=0}^{d-1} \delta(E-E_n).
\end{equation}
Here $\delta(E)$ is a Dirac delta mass centered at zero.
Given a state $|\vec{r}\rangle$, the Local Density of State (LDOS)
\begin{equation}
\hat{\rho}(E) = 
\sum_{k=1}^{n-1} \langle \vec{r} | \vec{u}_n \rangle \delta(E-E_n)
\end{equation}
is also of interest in many settings.
For example, if $|\vec{r}\rangle$ is a random state drawn from the uniform distribution on the unit hypersphere, then $\hat{\rho}(E)$ is an unbiased estimator for $\rho(E)$.
In fact, \emph{quantum typicality} \cite{goldstein_lebowitz_mastrodonato_tumulka_zanghi_10,jin_willsch_willsch_lagemann_michielsen_deraedt_21} ensures $\hat{\rho}(E)$ concentrates around $\rho(E)$.
It therefore often suffices to use $\hat{\rho}(E)$ as a proxy for $\rho(E)$, but in the case that the variance of a single sample is too large, one can sample multiple random states independently and then average the corresponding LDOSs to reduce the variance. \cite{alben_blume_krakauer_schwartz_75,skilling_89,jaklic_prelovsek_94,weisse_wellein_alvermann_fehske_06,schnack_richter_heitmann_richter_steinigeweg_20}
In numerical analysis this is called stochastic trace estimation and has been analyzed in detail.\cite{girard_87,girard_89,hutchinson_89,avron_toledo_11,roostakhorasani_ascher_14,cortinovis_kressner_21}

The aim of the KPM is to produce a ``density'' $\rho_{\KPM}(E)$ approximating the LDOS\footnote{If we can compute the moments of the DOS $\rho(E)$, then we can apply the KPM to the DOS directly.} $\hat{\rho}(E)$.
Towards this end, let $\sigma(E)$ be a fixed \emph{reference density} and expand $\hat{\rho}(E)/\sigma(E)$ as a formal polynomial series
\begin{equation}
    \frac{\hat{\rho}(E)}{\sigma(E)} 
    = \sum_{n=0}^{\infty} \mu_n p_n(E),
\end{equation}
where $\{p_n\}$ are the orthonormal polynomials with respect to $\sigma(E)$.
Using the orthonormality of the $\{p_n\}$ with respect to $\sigma(E)$, we can compute $\mu_n$ by
\begin{align}
    \mu_n
    &= \int \sigma(E) \frac{\hat{\rho}(E)}{\sigma(E)} p_n(E) \d E
    \\&= \int \hat{\rho}(E) p_n(E) \d E
    \\&= \langle \vec{r} | p_n(\vec{H}) | \vec{r} \rangle. \label{eqn:moment}
\end{align}
Thus, we see the $\{\mu_n\}$ are the so-called (modified) moments of $\hat{\rho}(E)$ with respect to the orthogonal polynomials of $\sigma(E)$ and can be computed without explicit knowledge of $\hat{\rho}(E)$ using the expression in \cref{eqn:moment}.

Computing the first $s$ moments naturally gives an approximation $\rho_{\KPM}(E)$ to $\hat{\rho}(E)$ defined by 
\begin{equation}
    \label{eqn:kpm}
    \rho_{\KPM}(E)
    = \sigma(E) \sum_{n=0}^{s} \mu_n p_n(E).
\end{equation}
When the degree of the approixmation $s\to\infty$, 
\begin{equation}
    \rho_{\KPM}(E)
    \to
    \sigma(E) \sum_{n=0}^{\infty} \mu_n p_n(E)
    = \hat{\rho}(E),
\end{equation}
and convergence is expected to be at a rate $O(s^{-1})$ in the Wasserstein distance .\cite{weisse_wellein_alvermann_fehske_06,braverman_krishnan_musco_22,chen_trogdon_ubaru_22}
A rigorous theoretical understanding of other types of convergence is important, particularly if a density approximation is desired. 
However, this is not entirely straightforward as $\hat{\rho}(E)$ is a linear combination of Dirac deltas, and therefore isn't even a density itself.

\subsection{Damping}

Strictly speaking, $\rho_{\KPM}(E)$ need not be a proper density as it may be negative for some values of $E$. 
This effect is particularly noticeable if $\hat{\rho}(E)/\sigma(E)$ is very spiky so that polynomial approximations have large Gibbs oscillations.
To ensure positivity, it is often suggested to use so \emph{damping kernels} which effectively result in an approximation
\begin{equation}
    \rho_{\KPM}(E)
    = \sigma(E) \sum_{n=0}^{s} g_n \mu_n p_n(E), 
\end{equation}
where the damping coefficients $\{g_n\}$ are carefully chosen. 
The most common choice of coefficients correspond to the so-called Jackson's damping kernel; see \cite{weisse_wellein_alvermann_fehske_06} for a detailed discussion on damping.
It is also possible to simply apply a standard convolution against the resulting approximation, although this does not necessarily ensure positivity.

\subsection{Evaluating orthogonal polynomials}

Assuming $\sigma(E)$ is a unit-mass (positive) density, the orthogonal polynomials satisfy a symmetric three term recurrence
\begin{equation}
    \label{eqn:op_recurrence}
    p_{n+1}(E) = \frac{1}{\delta_n}\big(E p_{n}(E) - \gamma_n p_n(E) - \delta_{n-1} p_{n-1}(E) \big)
\end{equation}
with initial conditions $p_{1}(E) = (1/\delta_0)(E p_{0}(E) - \gamma_0 p_0(E) )$, $p_0(E) = 1$,
for some set of recurrence coefficients $\{\gamma_n, \delta_n\}$ depending on $\sigma(E)$.
Throughout, we will assume that these coefficients are known (or can be computed).
Then, once the moments are known, the KPM approximation $\rho_{\KPM}(E)$ defined in \cref{eqn:kpm} can be obtained by evaluating the polynomials by this recurrence and then forming a linear combination of these polynomials.

\subsection{Choice of reference density}

The most common choice of reference density is 
\begin{equation} 
    \label{eqn:cheb_density}
\sigma_{\Eminh,\Emaxh}^{T}(E) 
= \frac{1}{\pi}\frac{1}{\sqrt{(\Emaxh-E)(E-\Eminh)}}, 
\end{equation} 
which is the orthogonality weight for the Chebyshev polynomials shifted and scaled from $[-1,1]$ to $[\Eminh,\Emaxh]$.
For this choice of reference density, there is an elegant and widely used algorithm for computing the moments which we summarize in \cref{sec:cheb_moments}.
In \cref{fig:support} we illustrate some of the impacts of the choice of $[\Eminh,\Emaxh]$ on the KPM approximation when this reference density is used. 
In particular, is is very important that $[\Eminh,\Emaxh]$ contains the spectrum of $\vec{H}$.
Thus, it is often suggested to take $[\Eminh,\Emaxh] = [\Emin-\eta,\Emax+\eta]$, where $\eta>0$, to avoid the risk of $[\Eminh,\Emaxh]$ not containing the entire spectrum.\cite{weisse_wellein_alvermann_fehske_06}
However, using nonzero values of $\eta$ reduces the resolution of the approach .\cite{schluter_gayk_schmidt_honecker_schnack_21}

The choice of reference density $\sigma(E)$ impacts the qualitative features of the KPM approximation $\rho_{\KPM}(E)$, and in principle, any choice of unit-mass density with finite moments is possible .\cite{silver_roder_94,weisse_wellein_alvermann_fehske_06}
First, observe that the KPM approximation is exact, if $\hat{\rho}(E) / \sigma(E)$ is a polynomial of degree at most $s$. 
Thus, one perspective is that $\rho(E)$ should be chosen to try and make $\hat{\rho}(E) / \sigma(E)$ as easy to approximate with polynomials as possible.
In particular, the support of $\hat{\rho}(E)$ must contain the support of $\sigma(E)$. 
The difficulty of such an approach is that many properties of the spectrum of $\hat{\rho}(E)$ are not known ahead of time. 
Even the most basic properties such as an interval $[\Eminh,\Emaxh]$ containing the support of $\hat{\rho}(E)$ are often unknown a priori and must be approximated numerically as a pre-processing step.
The approach we describe in \cref{sec:adaptive} addresses these difficulties by allowing $\sigma(E)$ to be chosen after computation has completed.

In our experiments, we will make use of reference densities of the form
\begin{equation}
\label{eqn:reference_sum}
    \sigma(E) 
    = \sum_i w_i \sigma_{\Eminh_i,\Emaxh_i}^{T}(E) .
\end{equation}
The recurrence coefficients for the orthogonal polynomials of distributions like \cref{eqn:reference_sum} are easily computed, since integrals of polynomials against each term can be computed exactly using quadrature rules.\cite{saad_83,geronimo_vanassche_88}



\subsection{Computing Chebyshev moments}
\label{sec:cheb_moments}

In the case that $\sigma(E) = \sigma_{\Eminh,\Emaxh}^{T}(E)$ is the Chebyshev density \cref{eqn:cheb_density}, the modified moments can be computed efficiently using properties of Chebyshev recurrences.
Such an approach is described in detail in the literature,\cite{skilling_89,silver_roder_94,weisse_wellein_alvermann_fehske_06} but as this is by far the most common approach to implementing the KPM, we provide a brief overview to put the contrubtions of this paper into context.

Recall the Chebyshev polynomials are defined by the recurrence
\begin{equation}
\label{eqn:cheb}
    T_{n+1}(E) = 2x T_{n}(E) - T_{n-1}(E),
\end{equation}
with initial conditions $T_1(E) = E$ and $T_0(E) = 1$.
These polynomials are orthogonal with respect to the weight $1/\sqrt{1-E^2}$.
The Chebyshev polynomials also satisfy the useful identities
\begin{align}
    T_{2n}(E) &= 2T_n(E)^2 - 1 \label{eqn:cheb2n}
    \\
    T_{2n+1}(E) &= 2T_{n+1}(E)T_n(E) - T_1(E). \label{eqn:cheb2n1}
\end{align}

The Chebyshev polynomials shifted and scaled from $[-1,1]$ to $[\Eminh,\Emaxh]$ are defined by
\begin{equation}
\label{eqn:shifted_cheb}
\tilde{T}_n(E) = T_n((E-\alpha)/\beta),
\end{equation}
where
\begin{equation}
\alpha = (\Emaxh+\Eminh)/2
    ,\qquad 
    \beta = (\Emaxh-\Eminh)/2.
\end{equation}
It's straightforward to see that the orthonormal polynomials of $\sigma_{\Eminh,\Emaxh}^{T}(E)$ are
\begin{equation}
\label{eqn:on_cheb}
    p_n(E) = \sqrt{2-\delta_{0,n}} \tilde{T}_n(E),
\end{equation}
where $\delta_{0,n}$ is the Kronecker delta.

In order to compute the moments, one can run the matrix version of the Chebyshev recurrence \cref{eqn:shifted_cheb},
\begin{equation}
    |\vec{v}_{n+1}\rangle
    =
    (2/\beta) (\vec{H}-\alpha\vec{I}) | \vec{v}_{n} \rangle - |\vec{v}_{n-1} \rangle,
\end{equation}
with initial conditions $\quad|\vec{v}_1\rangle = (1/\beta)(\vec{H}-\alpha\vec{I})|\vec{v}_0\rangle$ and $|\vec{v}_0\rangle = |\vec{r}\rangle$.
Then, at step $n$, we have $|\vec{v}_n\rangle = \tilde{T}_n(\vec{H})|\vec{r}\rangle$.
Using \cref{eqn:cheb2n,eqn:cheb2n1} we see the moments can then be computed by $\mu_0 = 1$, $\mu_1 = \sqrt{2} \langle \vec{r}|\vec{v}_1 \rangle$, and
\begin{align}
    \mu_{2n} &= 2 \sqrt{2} \langle \vec{v}_n | \vec{v}_n \rangle - \sqrt{2}  \mu_0
    \\
    \mu_{2n+1} &=  2 \sqrt{2} \langle \vec{v}_{n+1} | \vec{v}_n \rangle - \mu_1.
\end{align}
Note the factors of $\sqrt{2}$ are due to the fact that we are working with the orthonormal Chebyshev polynomials \cref{eqn:on_cheb} rather than the typical Chebyshev polynomials.

This approach is summarized in \cref{alg:cheb-moments} and clearly requires $\Eminh$ and $\Emaxh$ to be specified ahead of time.
If $[\Eminh,\Emaxh]$ does not contain all of the energies of $\vec{H}$, then the algorithm is \emph{exponentially unstable}. 
On the other hand if $[\Eminh,\Emaxh]$ is much wider than the energies of $\vec{H}$, then the convergence may be slowed.

\begin{algorithm}[H]
\caption{Chebyshev moments}
\label{alg:cheb-moments}
\fontsize{10}{10}\selectfont
\begin{algorithmic}[1]
\Procedure{Cheb-moments}{$\vec{H}, |\vec{r}\rangle, k, \Eminh,\Emaxh$}
\State $\alpha = (\Emaxh-\Eminh)/2$, $\beta = (\Emaxh+\Eminh)/2$
\State $|\vec{v}_0\rangle = |\vec{r}\rangle$, $|\vec{v}_1\rangle = (1/\beta)(\vec{H}-\alpha\vec{I})|\vec{v}_0\rangle$
\State $\mu_0 = 1$, $\mu_1 =\sqrt{2} \langle \vec{v}_1 |\vec{v}_0\rangle$
\For {\( n=1,2,\ldots,k-1 \)}
    \State \( | \vec{v}_{n+1} \rangle = (2/\beta) (\vec{H}-\alpha\vec{I}) | \vec{v}_{n} \rangle - |\vec{v}_{n-1} \rangle \)
    \State $\mu_{2n} = 2 \sqrt{2} \langle \vec{v}_n | \vec{v}_n \rangle - \sqrt{2} \mu_0$
    \State $\mu_{2n+1} = 2 \sqrt{2} \langle \vec{v}_{n+1} | \vec{v}_n \rangle - \mu_1$
\EndFor
\State \Return $\mu_0, \mu_1, \ldots, \mu_{2k}$ 
\EndProcedure
\end{algorithmic}
\end{algorithm}

\section{A Lanczos-based spectrum adaptive KPM}
\label{sec:adaptive}
We now describe our proposed algorithm with allows $\sigma(E)$ (including it's support) to be chosen after the expensive aspects of the algorithm have been carried out.
This allows the KPM approximation generated to adapt to the energy spectrum of the Hamiltonian $\vec{H}$.
A related and very general approach to obtaining quadrature approximations from moment data has recently been described .\cite{chen_trogdon_ubaru_22}
The approach in this paper is more focused/straightforward, as we focus only on implementing the KPM using Lanczos.
Indeed, our approach can be summarized on one sentence: \emph{use the output of Lanczos to compute the KPM moments}.

\subsection{The Lanczos algorithm}

When run on $\vec{H}$ and $|\vec{r}\rangle$ for $k$ iterations, the Lanczos algorithm (\cref{alg:lanczos}) iteratively produces an orthonormal basis $\{ |\vec{v}_n\rangle \}$ for the Krylov subspace 
\begin{equation}
    \label{eqn:krylov}
    \operatorname{span}\{|\vec{r}\rangle, \vec{H}|\vec{r}\rangle, \ldots, \vec{H}^k|\vec{r}\rangle\}.
\end{equation}
This is done via a symmetric three-term recurrence
\begin{equation}
    \label{eqn:krylov_recurrence}
    |\vec{v}_{n+1} \rangle
    = \frac{1}{\beta_{n}} \big( \vec{H} |\vec{v}_{n} \rangle
    - \alpha_{n} |\vec{v}_{n} \rangle 
    - \beta_{n-1} |\vec{v}_{n-1} \rangle \big) 
\end{equation}
with initial conditions
$|\vec{v}_1 \rangle = (1/\beta_{0}) ( \vec{H} |\vec{v}_{0} \rangle
    - \alpha_{0} |\vec{v}_{0} \rangle )$ and $|\vec{v}_0\rangle = |\vec{r}\rangle$.
At each step $\alpha_{n}$ is chosen so that $\langle \vec{v}_{n+1}|\vec{v}_{n}\rangle = 0$ and then $\beta_{n}$ is chosen so that $\langle \vec{v}_{n+1}|\vec{v}_{n+1}\rangle=1$.
In exact arithmetic, $|\vec{v}_{n+1}\rangle$ is automatically orthogonal to $|\vec{v}_i\rangle$ for all $i\leq n-2$ by symmetry. 
However, those familiar with the Lanczos algorithm in finite precision arithmetic may be skeptical that we have omitted any form of reorthogonalization.
We discuss the stability of our approach in finite precision arithmetic in \cref{sec:stability}, and argue that reorthogonalization is not needed.

After $k$ iterations of the Lanczos iteration, the recurrence coefficients form a $(k+1)\times (k+1)$ symmetric tridiagonal matrix
\begin{equation}
\label{eqn:Hk}
    \vec{H}_k = \operatorname{tridiag}
    \left(\hspace{-.75em} \begin{array}{c}
        \begin{array}{cccc} \beta_0 & \beta_1 & \cdots & \beta_{k-1} \end{array} \\
            \begin{array}{ccccc} \alpha_0 & \alpha_1 & \cdots& \alpha_{k-1} & 0 \end{array} \\
        \begin{array}{cccc} \beta_0 & \beta_1 & \cdots & \beta_{k-1} \end{array} 
    \end{array} \hspace{-.75em}\right).
\end{equation}
If we write the Lanczos basis as $\vec{V} = \sum_{n=0}^{k} |\vec{v}_n\rangle\langle \vec{e}_n |$, where $|\vec{e}_n\rangle$ is the all zeros vector with a one in index $n$, 
it is not hard to see \cref{eqn:krylov_recurrence} implies
\begin{equation}
    \label{eqn:krylov_brecurrence}
    \vec{H}\vec{V} = \vec{V} \vec{H}_k + |\vec{v}\rangle\langle \vec{e}_k|
\end{equation}
for some vector $|\vec{v}\rangle$. 
Note that it is somewhat more common to write such a recurrence with the upper-left $k\times k$ principle submatrix of $\vec{H}_k$ which is closely related to Gaussian quadrature .\cite{golub_meurant_09}
However, as will become apparent in the next section, using $\vec{H}_k$ as defined in \cref{eqn:Hk} will provide slightly more approximation power in our Lanczos-based KPM.

\begin{algorithm}[H]
\caption{Lanczos}
\label{alg:lanczos}
\fontsize{10}{10}\selectfont
\begin{algorithmic}[1]
\Procedure{Lanczos}{$\vec{H}, |\vec{r}\rangle, k, $}
\State $|\vec{v}_0 \rangle = |\vec{r}\rangle$
\State $|\tilde{\vec{v}}_1\rangle = \vec{H} |\vec{v}_0 \rangle$
\State $|\hat{\vec{v}}_1 = |\tilde{\vec{v}}_1\rangle - \alpha_0 |\vec{v}_0 \rangle$,~~$\alpha_0 = \langle \vec{v}_0| \tilde{\vec{v}}_1\rangle$
\State $|\vec{v}_1\rangle = \hat{\vec{v}}_1 / \beta_0$,~~ $\beta_0 = \langle \hat{\vec{v}}_1|\hat{\vec{v}}_1\rangle$
\For {\( n=1,2,\ldots,k-1 \)}
    \State $|\tilde{\vec{v}}_{n+1}\rangle = \vec{H} |\vec{v}_n \rangle - \beta_{n-1} | \vec{v}_{n-1} \rangle$
    \State $\alpha_n = \langle \vec{v}_{n}| \tilde{\vec{v}}_n\rangle$
    \State $|\hat{\vec{v}}_{n+1}\rangle = |\tilde{\vec{v}}_{n+1} \rangle - \alpha_n |\vec{v}_n \rangle$
    \State $\beta_{n} = \langle \hat{\vec{v}}_{n+1}|\hat{\vec{v}}_{n+1}\rangle$ \label{line:normalize}
    \State $|\vec{v}_{n+1}\rangle = |\hat{\vec{v}}_{n+1}\rangle /\beta_n$
\EndFor
\State \Return $\alpha_0, \alpha_1, \ldots, \alpha_{k-1}$, $\beta_0, \beta_1, \ldots, \beta_{k-1}$ 
\EndProcedure
\end{algorithmic}
\end{algorithm}

\subsection{Getting the KPM moments}
It is well-known that the Lanczos tridiagonal matrix $\vec{H}_k$ contains information suitable for computing the polynomial moments of $\hat{\rho}(E)$ through degree $2k$.
\begin{theorem}\label{thm:poly_exact}
Let $p$ be any polynomial of degree at most $2k$. 
Then
    \begin{equation}
        \langle \vec{r} | p(\vec{H}) | \vec{r}\rangle 
        = \int \hat{\rho}(E) p(E) \d E
        = \langle \vec{e}_0 | p(\vec{H}_k) | \vec{e}_0 \rangle.
    \end{equation}
\end{theorem}

\begin{proof}
    The first equality is by definition of the LDOS $\hat{\rho}(E)$.
    Due to linearity of polynomials, it suffices to show 
    \begin{equation}
    \label{eqn:func_exact}
        \langle \vec{r} | \vec{H}^n | \vec{r} \rangle = \langle \vec{e}_0 | \vec{H}_k^n | \vec{e}_0 \rangle
    \end{equation}
    for $n=0,1,\ldots, 2k$.
    We will first show that $\vec{H}^n |\vec{r}\rangle = \vec{V} \vec{H}_k^n |\vec{e}_0\rangle$ for all $n\leq k$.
    Since $\vec{V}^\dagger\vec{V} = \vec{I}$, this immediately implies the desired result.

    Suppose $\vec{H}^{n-1} |\vec{r}\rangle = \vec{V} \vec{H}_k^{n-1} |\vec{e}_0\rangle$.
    Then, since $|\vec{r}\rangle = \vec{V}|\vec{e}_0\rangle$, we can use \cref{eqn:krylov_brecurrence} to write
    \begin{equation}
        \vec{H}^n |\vec{r}\rangle
        = \vec{H}\vec{V} \vec{H}_k^{n-1} |\vec{e}_0\rangle
        = \vec{V} \vec{H}_k^n |\vec{e}_0\rangle + |\vec{v}\rangle \langle \vec{e}_k|\vec{H}_k^n|\vec{e}_0\rangle.
    \end{equation}
    Since $\vec{H}_k$ is tridiagonal, $\vec{H}_k^n$ has bandwidth $2n+1$ and $        \langle \vec{e}_k|\vec{H}_k^n|\vec{e}_0\rangle = 0$
    provided $n \leq k$.
    The base case $|\vec{r}\rangle = \vec{V} |\vec{e}_0\rangle$ is trivial.
\end{proof}

The critical observation is that this allows us to obtained the KPM moments $\{\mu_n\}$ with respect to a reference density $\sigma(E)$ which we can choose \emph{after} we have run the Lanczos computation.
In fact, we can cheaply produce approximations corresponding to various different reference measures as this process no longer involves computations with $\vec{H}$ or any vectors of length $d$.

If we choose $\sigma(E) = \sigma_{\Eminh,\Emaxh}^{T}(E)$, we can compute the moments by applying \cref{alg:cheb-moments} to $\vec{H}_k$ and $|\vec{e}_0\rangle$. 
Of note is the fact that the Lanczos algorithm produces high accuracy estimates of extremal eigenvalues when started on a random vector.\cite{kuczyski_wozniakowski_92,zhou_li_11,martinsson_tropp_20}
In particular, the largest and smallest eigenvalues of the top-left $k\times k$ sub-matrix of $\vec{H}_k$ approximate those of $\vec{H}$ from the interior.
We can use this to help guide our choice of $\Eminh$ and $\Emaxh$.

For other choices of $\sigma(E)$, we can compute the moments directly via the three term-recurrence for the orthogonal polynomials. 
In particular, we use the matrix version of \cref{eqn:op_recurrence} 
\begin{equation}
    |\vec{u}_{n+1} \rangle
    = \frac{1}{\delta_{n}} \big( \vec{H}_k |\vec{u}_{n} \rangle
    - \gamma_{n} |\vec{u}_{n} \rangle 
    - \delta_{n-1} |\vec{u}_{n-1} \rangle \big) 
\end{equation}
with initial conditions $|\vec{u}_{1} \rangle = (1/\delta_{0}) ( \vec{H}_k |\vec{u}_{0} \rangle - \gamma_{n} |\vec{u}_{0} \rangle )$ and $|\vec{u}_0\rangle = |\vec{e}_0\rangle$.
Then $|\vec{u}_{n}\rangle = | \vec{p}_{n}(\vec{H}_k) |\vec{e}_0 \rangle$, so we can compute $\mu_n = \langle \vec{e}_0 | \vec{u}_n \rangle$ for $n\leq 2k$.
This is summarized in \cref{alg:lanczos_moments}.
The cost of this process depends only on $k$ and not on $d$, so tricks for halving the number of matrix-vector products with $\vec{H}_k$ are not needed.

\begin{algorithm}[H]
\caption{Get KPM moments from Lanczos}
\label{alg:lanczos_moments}
\fontsize{10}{10}\selectfont
\begin{algorithmic}[1]
\Procedure{moments-from-Lanczos}{$\vec{H}_k, \sigma(E)$}
\State Obtain recurrence coefficients $\{\gamma_n,\delta_n\}$ for the orthogonal polynomials of $\sigma(E)$
\State $|\vec{u}_0 \rangle = |\vec{e}_0\rangle$, $|\vec{u}_1 \rangle = (1/\delta_0) (\vec{H}_k|\vec{u}_0\rangle - \gamma_0 \vec{u}_0\rangle) $
\State $\mu_0 = \langle \vec{e}_0 | \vec{u}_0 \rangle$,  $\mu_1 = \langle \vec{e}_0 | \vec{u}_1 \rangle$
\For {\( n=1,2,\ldots,2k-1 \)}
    \State $|\vec{u}_{n+1} \rangle
    = (1/\delta_{n}) ( \vec{H}_k |\vec{u}_{n} \rangle
    - \gamma_{n} |\vec{u}_{n} \rangle 
    - \delta_{n-1} |\vec{u}_{n-1} \rangle )$
    \State $\mu_{n+1} = \langle \vec{e}_0 | \vec{u}_{n+1} \rangle$
\EndFor
\State \Return $\mu_0, \mu_1, \ldots, \mu_{2k}$ 
\EndProcedure
\end{algorithmic}
\end{algorithm}

\subsection{Stability of the Lanczos algorithm}
\label{sec:stability}

The Lanczos algorithm is unstable in the sense that the tridiagonal matrix $\vec{H}_k$ and basis vectors $\{|\vec{v}_n\rangle\}$ produced in finite precision arithmetic may be nothing like what would be obtained in exact arithmetic.
In particular, by symmetry, $|\vec{v}_{n+1}\rangle$ is automatically orthogonal to $|\vec{v}_0\rangle, \ldots, |\vec{v}_{n-1}\rangle$ in exact arithmetic.
However, in finite precision arithmetic, this is not longer even approximately true and the Lanczos basis vectors produced can completely lose orthogonality and even linear independence.
To fix this, it is common to use \emph{reorthogonalization}; that is, to explicitly orthogonalize $|\hat{\vec{v}}_{n+1}$ against  $|\vec{v}_0\rangle, \ldots, |\vec{v}_{n-1}\rangle$ before normalizing in \cref{line:normalize} of \cref{alg:lanczos}.
This of course drastically increases the computational requirements to be able to run the algorithm; in particular, the memory required for reorthogonalization scales as $O(dk)$ and the arithmetic cost as $O(dk^2)$.

Despite the instabilities of the Lanczos algorithm without reorthogonalization, there is significant theoretical \cite{paige_70,paige_76,paige_80,greenbaum_89,strakos_greenbaum_92,druskin_knizhnerman_91} and empirical \cite{long_prelovsek_sawish_karadamoglou_03,schnack_richter_steinigeweg_20,chen_trogdon_ubaru_21} evidence that the Lanczos algorithm is highly effective for tasks related to density of states approximation.
In fact, while not widely known, the Lanczos algorithm is \emph{forward stable} for the tasks of computing Chebyshev polynomials and moments .\cite{druskin_knizhnerman_91,knizhnerman_96}
We summarize the high-level ideas behind these works, the results of which we believe are relevant to the computational physics and chemistry communities.

In finite precision arithmetic, the outputs of the Lanczos algorithm satisfy a perturbed version of \cref{eqn:krylov_brecurrence},
\begin{equation}
    \label{eqn:krylov_brecurrence_fp}
    \vec{H}\vec{V} = \vec{V} \vec{H}_k + |\vec{v}\rangle\langle \vec{e}_k|
    + \vec{F}.
\end{equation}
Here, $\vec{F}$ accounts for local errors which can be expected to be on the size of machine precision.
In addition, while $\vec{V}$ need not be orthonormal, $|1-\langle \vec{v}_{n+1}|\vec{v}_{n+1}\rangle|$ and $|\langle \vec{v}_{n+1}|\vec{v}_n\rangle|$ are on the order of machine precision. 
This (and much more) is analyzed in detail in .\cite{paige_70,paige_76,paige_80}
These assumptions form the basis of essentially all analyses of the behavior of Lanczos-based methods in finite precision arithmetic.

Following the proof of \cref{thm:poly_exact}, we might try to use \cref{eqn:krylov_brecurrence_fp} to understand the difference between $\vec{H}^n |\vec{r}\rangle$ and $\vec{V} \vec{H}_k^n|\vec{e}_0\rangle$; i.e. how closely \cref{eqn:func_exact} holds in finite precision arithmetic.
However, it's not hard to see that this results in an error term with an exponential dependence on $k$.
This is fundamentally because the monomial basis is very poorly conditioned and therfore not a good choice to work with numerically.
Indeed, if instead we use a Chebyshev basis, then it can be shown the error term itself satisfies Chebyshev-like three-term recurrence and grows only polynomially with $k$.
This yields the following bound:
\begin{theorem}[{Druskin Knizhnerman 1991 \cite{druskin_knizhnerman_91}, informal}]\label{thm:dk91}
Suppose $[\Eminh,\Emaxh] = [\Emin - \eta,\Emax + \eta]$ for some $\eta = O(\epsilon_{\textup{mach}} \operatorname{poly}(k))$.
    Let $\vec{H}_k$ be the output of the Lanczos algorithm run on $\vec{H}$ and $|\vec{r}\rangle$ for $k$ iterations in finite precision arithmetic without reorthogonalization. 
    Then, for any $n\leq 2k$, 
    \begin{equation}
        \big\| \tilde{T}_n(\vec{H}) | \vec{r} \rangle
        - \vec{V} \tilde{T}_n(\vec{H}_k) | \vec{e}_0 \rangle \big\|
        = O\big(\epsilon_{\textup{mach}} \operatorname{poly}(k)\big).
\end{equation}
\end{theorem}
Here $\tilde{T}_n$ is as in \cref{eqn:shifted_cheb} and the big-$O$ hides mild dimensional constants and dependencies on $\Eminh$ and $\Emaxh$.
The actual statement of \cite[Theorem 1]{druskin_knizhnerman_91} is more precise, and gives explicit bounds for the the degree of the $\operatorname{poly}(k)$ terms. 
We remark that in numerical analysis, the precise numerical value of bounds is often much less important than the intuition conveyed by the bound.

With some additional knowledge of properties satisfied by $\vec{V}^\dagger \vec{V}$ in finite precision arithmetic (which can be very far from the identity), a similar result can be shown for the Chebyshev moments:
\begin{theorem}[{Knizhnerman 1996 \cite{knizhnerman_96}, informal}]\label{thm:kn96}
    Under similar assumptions to \cref{thm:dk91}, for any $n\leq 2k$, 
    \begin{equation}
        \big| \langle \vec{r} | \tilde{T}_n(\vec{H}) | \vec{r} \rangle
        - \langle \vec{e}_0 | \tilde{T}_n(\vec{H}_k) | \vec{e}_0 \rangle \big|
        = O\big(\epsilon_{\textup{mach}} \operatorname{poly}(k)\big).
\end{equation}
\end{theorem}
\Cref{thm:kn96} implies that the (appropriately scaled) Chebyshev moments can be obtained from the matrix $\vec{H}_k$, even when Lanczos was carried out in finite precision arithmetic.
Thus, the KPM with $\sigma_{\Eminh,\Emaxh}^{T}(E)$ can be implemented from Lanczos after $\vec{H}_k$ has been obtained, even in finite precision arithmetic.
As with \cref{thm:dk91}, \cite[Theorem 1]{knizhnerman_96} is much more precise than what is stated in \cref{thm:kn96}.

Note that any polynomial $p(x)$ of degree $2k$ can be decomposed in a Chebyshev series
\begin{equation}
    p(E) 
    = \sum_{n=0}^{2k} c_n \tilde{T}_n(E).
\end{equation}
Moreover, if $|p(E)|\leq M$ for all $E\in [\Eminh,\Emaxh]$, then
\begin{equation}
    |c_i|
    = \left| 
    (2-\delta_{0,n})\int \sigma_{a,b}^{T}(E) p(E) \tilde{T}_n(E) \d E \right|
    \leq 2 M.
\end{equation}
Applying the triangle inequality gives the bound
    \begin{equation}
        \big| \langle \vec{r} | p(\vec{H}) | \vec{r} \rangle
        - \langle \vec{e}_0 | p(\vec{H}_k) | \vec{e}_0 \rangle \big|
        = O\big(M \epsilon_{\textup{mach}} \operatorname{poly}(k)\big).
\end{equation}
In other words, \cref{thm:kn96} can be upgraded to hold for any bounded polynomial.
Thus, the KPM can also be implemented for choices of $\sigma(E)$ whose orthogonal polynomials are well-behaved.

We remark this that a similar argument also implies that, even in finite precision arithmetic, the FTLM approximation to $\langle \vec{r} | f(\vec{H}) | \vec{r} \rangle$ is accurate provided $f(E)$ has a good polynomial approximation; i.e. the same sort of result as would be expected in exact arithmetic.
A more detailed analysis is found in .\cite{knizhnerman_96}
This provides a theoretical justification to the observation that the FTLM works well, even without reorthogonalization .\cite{schnack_richter_steinigeweg_20}

\subsection{Computational costs}

The overall computational costs of our energy adaptive KPM and the standard KPM (assuming $[\Eminh,\Emaxh]$ is known) are almost identical.
In addition to the storage required for $\vec{H}$, both algorithms require storage for 3 vectors of length $d$.
At iteration $n$, the Lanczos algorithm (without reorthogonalization) and a standard implementation of the KPM require one matrix-vector product, several vector updates, and two inner products.
The algorithms also have additional lower-order arithmetic and storage costs which depend only on the maximum number of iterations $k$ but not the dimension $d$.
Assuming $k\ll d$, these costs are negligible.

As noted above, the standard KPM typically requires a pre-processing step in which a suitable interval $[\Eminh,\Emaxh]$ is determined, often via a few iterations of Lanczos.
While our Lanczos-based KPM avoids the need for such a pre-processing step, this pre-processing step is often cheap relative to the overall computation.
In such cases, the fact that our algorithm avoids this step is not particularly significant from a runtime perspective.

Finally, we note one situation in which the runtimes of the two algorithms may differ is on high-performance clusters where the time spent on communication for inner products and can dominant the time spent on arithmetic computation.
Indeed, in the case of the KPM, the two inner products are used to compute $\mu_n$ and $\mu_{n+1}$ and do not prevent the algorithm from proceeding.
On the other hand, the two inner products in  Lanczos  are used to compute $\alpha_n$ and $\beta_n$ are blocking, and therefore must be completed before the algorithm can proceed.

In such settings, if the cost of the pre-processing step is significant, one could run the energy-adaptive KPM suggested here for several iterations to determine good choices of $[\Eminh,\Emaxh]$, and $\sigma(E)$. 
Assuming the Lanczos basis vectors are stored, then $p_n(\vec{H})|\vec{r}\rangle$ and $p_{n-1}(\vec{H})|\vec{r}\rangle$ can be computed without any additional matrix vector products at which point an explicit three-term recurrence can be continued without the need for blocking inner products.
We leave such implementation details to practitioners who have better knowledge of their individual computing environments.

\section{Numerical Experiments}

In this section we provide several numerical examples to demonstrate the potential usefulness of our spectrum adaptive KPM.
Rather than focusing on any single domain area, we aim to provide a diverse collection of examples from a variety of applications.
Each of these examples demonstrates a particular aspect of the paradigm of decoupling computation from approximation which may prove useful in practical settings.
Unless stated otherwise, all KPM approximations are computed using our Lanczos-based approach.

\subsection{Spin systems}

One of the main uses of KPM and related algorithms is in the study of thermodynamic properties of Heisenberg spin systems. \cite{morita_tohyama_20,schnack_richter_steinigeweg_20}
Here we consider the simplest example: the 1D XX spin chain of length $m$ with Hamiltonian
\[
\vec{H} = 
J\sum_{i=1}^{m-1}\left( \vec{\sigma}_i^{\textup{x}} \vec{\sigma}_{i+1}^{\textup{x}} + \vec{\sigma}_i^{\textup{y}} \vec{\sigma}_{i+1}^{\textup{y}} \right)
+ h \sum_{i=1}^{m} \vec{\sigma}_i^{\textup{z}}.
\]
This system is exactly solvable, meaning that the true spectrum can be computed analytically. \cite{karabach_97}
For our numerical experiments, we set $m=20$ so that $d = 2^{20} \approx 10^6$ and use $J=1/6$ and $h=6$.

\Cref{fig:support} shows a histogram of the exact spectrum, along with the degree $s=500$ KPM approximation of the LDOS corresponding to a single random state $|\vec{r}\rangle$.
The support $[\Eminh,\Emaxh]$ of the KPM approximation is varied to study the impact of estimating $[\Emin,\Emax]$.
If $E_n \not\in [\Eminh,\Emaxh]$ for some index $n$, then the KPM approximation deteriorates, but if $[\Eminh,\Emaxh]$ is taken too large, then convergence is slowed. 
Thus, reasonably accurate estimates of $\Emin$ and $\Emax$ are required.
Our spectrum adaptive KPM allows these estimates to be determined after computation with $\vec{H}$ has occurred.

\begin{figure}[htb]
    \includegraphics[scale=.6]{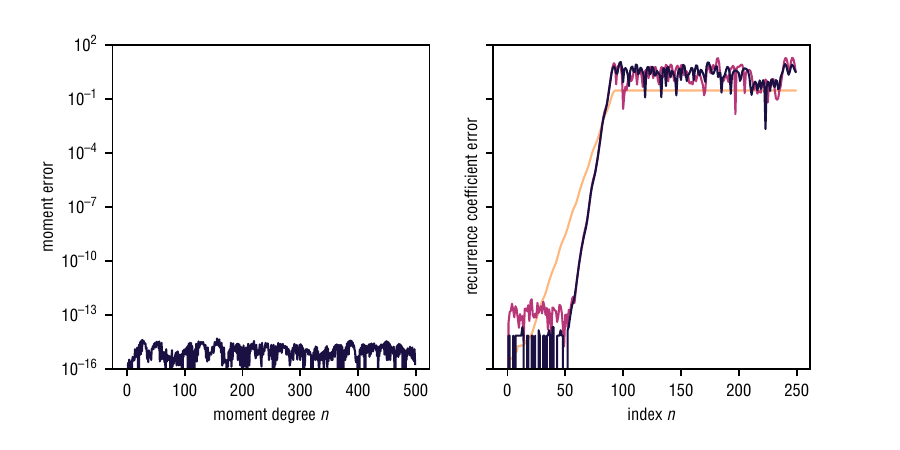}
    \caption{Study of relevant quantities in finite precision arithmetic.
    \emph{Legend}: 
    Left: error between Chebyshev moments computed directly and using the finite precision Lanczos recurrence and \cref{thm:poly_exact}.
    Right: 
    level of orthogonality of the Lanczos basis vectors $\max_{i<j\leq n}|\langle \vec{v}_i|\vec{v}_j\rangle|$
    ({\protect\raisebox{0mm}{\protect\includegraphics[scale=.7]{l0.pdf}}}) and
    error in Lanczos coefficients with and without reorthogonalization $|\alpha_n - \alpha_n^* |$ 
    ({\protect\raisebox{0mm}{\protect\includegraphics[scale=.7]{l1.pdf}}})
    and $|\beta_n - \beta_n^* |$
    ({\protect\raisebox{0mm}{\protect\includegraphics[scale=.7]{l2.pdf}}}).
    \emph{Takeaway}: Even though the Lanczos algorithm completely lost orthogonality, the Chebyshev moments are computed stably!
    This enables us to stably implement the Chebyshev KPM with Lanczos, even without reorthogonalzation.
    }
    \label{fig:moments_err}
\end{figure}

In \cref{fig:moments_err} we study the impact of finite precision arithmetic. 
We first consider how accurately the moments $\mu_n$ can be computed. 
Specifically, we compare our Lanczos-based algorithm and a standard implementation of the KPM and observe that the computed moments agree to essentially machine precision. 
This is expected due to \cref{thm:kn96}.
For reference, we also show the orthogonality of the Lanczos vector without reorthogonalization as well as the difference between the Lanczos recurrence coefficients $\{\alpha_n,\beta_n\}$ and what would have been obtained with reorthogonalization $\{\alpha_n^*,\beta_n^*\}$.
There is a complete loss of orthogonality and the coefficients obtained with and without reorthogonalization are vastly different. 
This implies that the moments agreeing is not simply due to the outputs with and without reorthogonalization being similar.

\subsection{Tight binding models}

Another common use of the KPM is in tight binding models.
Here we consider a cuboid section of a Zincblende crystal with nearest neighbor hopping.
The Hamiltonian and visualization of the zincblende crystal in \cref{fig:graphene} were generated using the Kwant code .\cite{groth_wimmer_akhmerov_waintal_14}
The resulting Hamiltonian is of dimension $d=56000$, and we output the average of 10 LDOSs corresopnding to random independent samples of $|\vec{r}\rangle$.

The DOS hs a large spike at zero, and this spike causes an undamaped Chebyshv KPM to exhibit massive Gibbs oscillations. 
These oscillations can be mitigated somewhat through the use of a damping kernel.
However, as seen in \cref{fig:graphene}, the resolution is only $O(1/k)$ and therefore $s$ must be taken large to get high resolution.

\begin{figure}[b]
    \includegraphics[scale=.6]{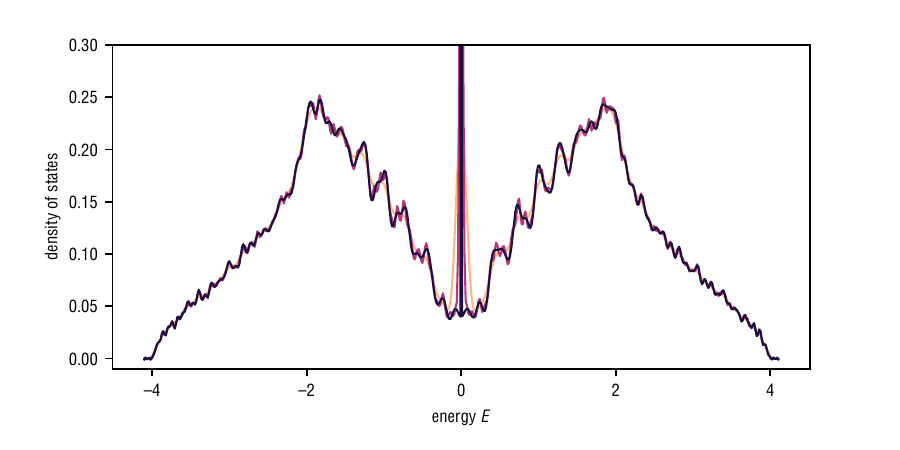}
    \caption{
    Approximation of DOS of a cubic Zincblende crystal with nearest neighbor hopping.
    \emph{Legend}: 
    reference density \cref{eqn:graphene_density} with $s=240$
    ({\protect\raisebox{0mm}{\protect\includegraphics[scale=.7]{l2.pdf}}}), 
    Chebyshev KPM with Jackson's damping with
    $s=240$
    ({\protect\raisebox{0mm}{\protect\includegraphics[scale=.7]{l0.pdf}}}) 
    and $s=800$
    ({\protect\raisebox{0mm}{\protect\includegraphics[scale=.7]{l1.pdf}}}).
    \emph{Takeaway}:
    Large spikes in the spectrum are hard to resolve with regular Chebyshev KPM, but can be resolved with a suitable choice of reference density $\sigma(E)$.
    Our spectrum adaptive KPM allows such an approximation to be computed without prior knowledge of the spectrum.    
    }
\label{fig:graphene}
\end{figure}

Instead, we can construct a reference density which is adaptive to this spike.
In particular, define
\begin{equation}
    \label{eqn:graphene_density}
    \sigma(E) = 0.05 \sigma_{1-\eta,1+\eta}^{T}(E) + 0.95\sigma_{\Eminh,\Emaxh}^{T}(E)
\end{equation}
where $\eta = 10^{-2}$.
Note that the relative weighting of the spike and the bulk spectrum, as well as the width of the spike can be tuned using our energy adaptive KPM. 
When properly chosen, this in higher resolution in other parts of the spectrum as the KPM approximation does not have to use as much of it's approximation power on approximating the spike in $\hat{\rho}(E)/\sigma(E)$.
Moreover, even without damping, Gibbs oscillations are relatively minor.

\subsection{Density Functional Theory}

In this example, we consider a matrix obtained in the study of $\text{Ga}_{41}\text{As}_{41}\text{H}_{72}$ with the pseudo-potential algorithm for real-space electronic structure calculation (PARSEC) .\cite{kronik_makmal_tiago_alemany_jain_huang_saad_chelikowsky_06}
The matrix \texttt{Ga41As41H72} can be obtained from the Sparse Matrix collection \cite{davis_hu_11} and has been used as a test matrix in past numerical work.\cite{zhou_li_11,li_xi_erlandson_saad_19}
This matrix is of dimension $d=268096$ and has many low-lying eigenvalues $|E_n|<100$ and a cluster of 123 large eigenvalues $E_n \in [1299,1301]$.
Thus, while the spectrum can be contained in two reasonably sized intervals, any single interval containing the spectrum must be large.
In this experiment, we output the average of 10 random LDOSs.

If the standard Chebyshev KPM is used, the zero part of $\hat{\rho}(E)/\sigma(E)$ in the gap must be approximated with a polynomial \footnote{If the number of outlying eigenvalues is known, it is possible to deflate them.\cite{weisse_wellein_alvermann_fehske_06,morita_tohyama_20} However, this requires additional computation including the storage of the eigenvectors (which may be intractable).}.
This significantly slows convergence, and even with $s=800$ moments the fine-grained structure of the upper spectrum is not resolved.
To avoid this delay of convergence, we can take $\sigma(E)$ as a density supported on two disjoint intervals $[\Eminh_1,\Emaxh_1]$ and $[\Eminh_2,\Emaxh_2]$ containing the spectrum of $\vec{H}$.
In particular, we take
\begin{equation}
\label{eqn:parsec_density}
    \sigma(E) = 0.95 \sigma_{\Eminh_1,\Emaxh_1}^{T}(E) + 0.05 \sigma_{\Eminh_2,\Emaxh_2}^{T}(E),
\end{equation}
where the intervals $[\Eminh_1,\Emaxh_1]$ and $[\Eminh_2,\Emaxh_2]$ are computed based on the eigenvalues of $\vec{H}_k$.
As seen in \cref{fig:parsec}, this provide higher resolution in each interval and the structure of the upper cluster is visible.
Here we have applied a simple convolutional filter to the KPM approximation of the upper eigenvalues approximation to reduce Gibbs oscillations.

\begin{figure}[htb]
    \includegraphics[scale=.6,trim=0 0 0 12,clip]{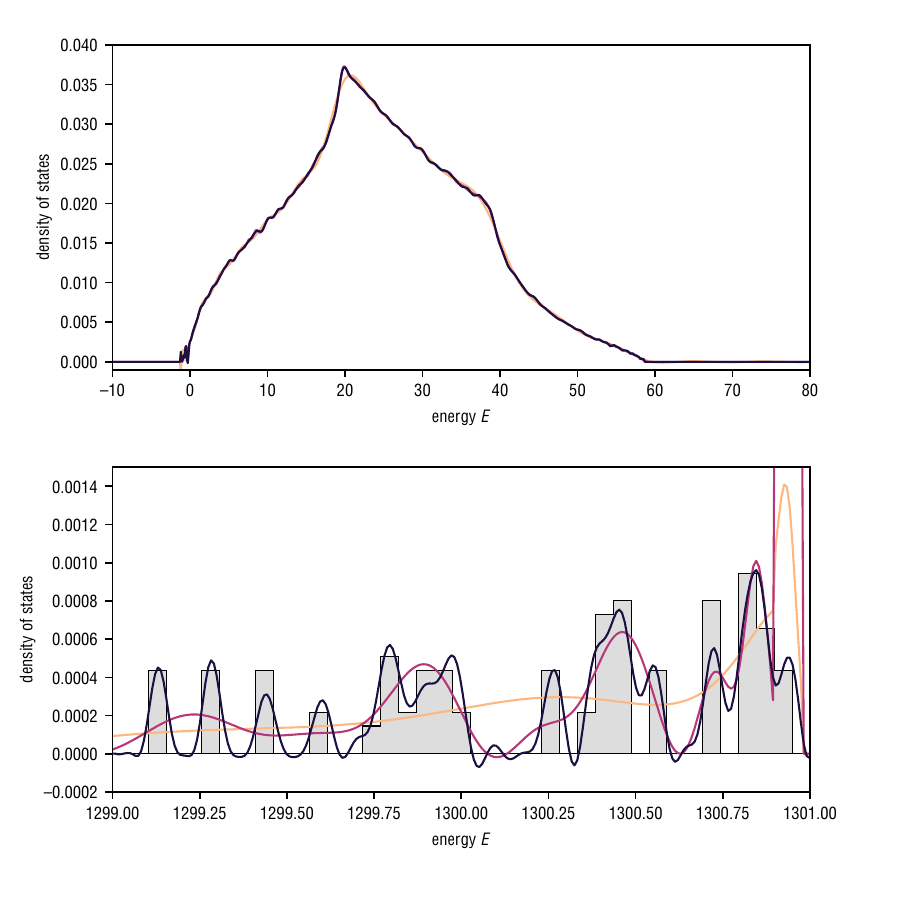}
    \caption{
    Approximation of DOS with a large gap in the spectrum.
    \emph{Legend}: 
    reference density \cref{eqn:parsec_density} with a convolutional filter and $s=200$
    ({\protect\raisebox{0mm}{\protect\includegraphics[scale=.7]{l2.pdf}}}), 
    Chebyshev KPM with Jackson's damping
    $s=200$
    ({\protect\raisebox{0mm}{\protect\includegraphics[scale=.7]{l0.pdf}}}) and
    $s=800$
    ({\protect\raisebox{0mm}{\protect\includegraphics[scale=.7]{l1.pdf}}}),
    histogram of top eigenvalues
    ({\protect\raisebox{0mm}{\protect\includegraphics[scale=.7]{hist.pdf}}}).
    \emph{Takeaway}:
    Even with 4 times the computation, the standard Chebyshev KPM does match the resolution of KPM with a more suitable choice of reference density $\sigma(E)$.
    Our spectrum adaptive KPM allows such an approximation to be computed without prior knowledge of the spectrum.
    }
    \label{fig:parsec}
\end{figure}

\section{Conclusion}

We have described an energy adaptive Kernel Polynomial Method based on the Lanczos algorithm. 
Our approach allows many different KPM approximations to be tested out for close to zero cost, after computation with $\vec{H}$ has finished.
Experiments demonstrate situations in which this allows the reference density $\sigma(E)$ to be chosen in such a way to improve the resolution of the approximation.
It is our belief that the paradigm of separating computation from the desired approximation is beneficial in most settings in which the KPM is used, and that our algorithm has the potential to improve the usability of the KPM.

\section*{Author  Declarations}

\subsection*{Conflict of interest}
\noindent
The authors have no conflicts to disclose.

\subsection*{Data availability}
\noindent
The data that support the findings of this study are available within the article.

\bibliography{refs.bib}

\begin{thebibliography}{54}%
\makeatletter
\providecommand \@ifxundefined [1]{%
 \@ifx{#1\undefined}
}%
\providecommand \@ifnum [1]{%
 \ifnum #1\expandafter \@firstoftwo
 \else \expandafter \@secondoftwo
 \fi
}%
\providecommand \@ifx [1]{%
 \ifx #1\expandafter \@firstoftwo
 \else \expandafter \@secondoftwo
 \fi
}%
\providecommand \natexlab [1]{#1}%
\providecommand \enquote  [1]{``#1''}%
\providecommand \bibnamefont  [1]{#1}%
\providecommand \bibfnamefont [1]{#1}%
\providecommand \citenamefont [1]{#1}%
\providecommand \href@noop [0]{\@secondoftwo}%
\providecommand \href [0]{\begingroup \@sanitize@url \@href}%
\providecommand \@href[1]{\@@startlink{#1}\@@href}%
\providecommand \@@href[1]{\endgroup#1\@@endlink}%
\providecommand \@sanitize@url [0]{\catcode `\\12\catcode `\$12\catcode
  `\&12\catcode `\#12\catcode `\^12\catcode `\_12\catcode `\%12\relax}%
\providecommand \@@startlink[1]{}%
\providecommand \@@endlink[0]{}%
\providecommand \url  [0]{\begingroup\@sanitize@url \@url }%
\providecommand \@url [1]{\endgroup\@href {#1}{\urlprefix }}%
\providecommand \urlprefix  [0]{URL }%
\providecommand \Eprint [0]{\href }%
\providecommand \doibase [0]{https://doi.org/}%
\providecommand \selectlanguage [0]{\@gobble}%
\providecommand \bibinfo  [0]{\@secondoftwo}%
\providecommand \bibfield  [0]{\@secondoftwo}%
\providecommand \translation [1]{[#1]}%
\providecommand \BibitemOpen [0]{}%
\providecommand \bibitemStop [0]{}%
\providecommand \bibitemNoStop [0]{.\EOS\space}%
\providecommand \EOS [0]{\spacefactor3000\relax}%
\providecommand \BibitemShut  [1]{\csname bibitem#1\endcsname}%
\let\auto@bib@innerbib\@empty
\bibitem [{\citenamefont {Skilling}(1989)}]{skilling_89}%
  \BibitemOpen
  \bibfield  {author} {\bibinfo {author} {\bibfnamefont {J.}~\bibnamefont
  {Skilling}},\ }\bibfield  {title} {\enquote {\bibinfo {title} {The
  eigenvalues of mega-dimensional matrices},}\ }in\ \href
  {https://doi.org/10.1007/978-94-015-7860-8_48} {\emph {\bibinfo {booktitle}
  {Maximum Entropy and Bayesian Methods}}}\ (\bibinfo  {publisher} {Springer
  Netherlands},\ \bibinfo {year} {1989})\ pp.\ \bibinfo {pages}
  {455--466}\BibitemShut {NoStop}%
\bibitem [{\citenamefont {Jakli{\v{c}}}\ and\ \citenamefont
  {Prelov{\v{s}}ek}(1994)}]{jaklic_prelovsek_94}%
  \BibitemOpen
  \bibfield  {author} {\bibinfo {author} {\bibfnamefont {J.}~\bibnamefont
  {Jakli{\v{c}}}}\ and\ \bibinfo {author} {\bibfnamefont {P.}~\bibnamefont
  {Prelov{\v{s}}ek}},\ }\bibfield  {title} {\enquote {\bibinfo {title}
  {{L}anczos method for the calculation of finite-temperature quantities in
  correlated systems},}\ }\href {https://doi.org/10.1103/physrevb.49.5065}
  {\bibfield  {journal} {\bibinfo  {journal} {Physical Review B}\ }\textbf
  {\bibinfo {volume} {49}},\ \bibinfo {pages} {5065--5068} (\bibinfo {year}
  {1994})}\BibitemShut {NoStop}%
\bibitem [{\citenamefont {Bai}, \citenamefont {Fahey},\ and\ \citenamefont
  {Golub}(1996)}]{bai_fahey_golub_96}%
  \BibitemOpen
  \bibfield  {author} {\bibinfo {author} {\bibfnamefont {Z.}~\bibnamefont
  {Bai}}, \bibinfo {author} {\bibfnamefont {G.}~\bibnamefont {Fahey}},\ and\
  \bibinfo {author} {\bibfnamefont {G.}~\bibnamefont {Golub}},\ }\bibfield
  {title} {\enquote {\bibinfo {title} {Some large-scale matrix computation
  problems},}\ }\href {https://doi.org/10.1016/0377-0427(96)00018-0} {\bibfield
   {journal} {\bibinfo  {journal} {Journal of Computational and Applied
  Mathematics}\ }\textbf {\bibinfo {volume} {74}},\ \bibinfo {pages} {71--89}
  (\bibinfo {year} {1996})}\BibitemShut {NoStop}%
\bibitem [{\citenamefont {Lin}, \citenamefont {Saad},\ and\ \citenamefont
  {Yang}(2016)}]{lin_saad_yang_16}%
  \BibitemOpen
  \bibfield  {author} {\bibinfo {author} {\bibfnamefont {L.}~\bibnamefont
  {Lin}}, \bibinfo {author} {\bibfnamefont {Y.}~\bibnamefont {Saad}},\ and\
  \bibinfo {author} {\bibfnamefont {C.}~\bibnamefont {Yang}},\ }\bibfield
  {title} {\enquote {\bibinfo {title} {Approximating spectral densities of
  large matrices},}\ }\href {https://doi.org/10.1137/130934283} {\bibfield
  {journal} {\bibinfo  {journal} {{SIAM} Review}\ }\textbf {\bibinfo {volume}
  {58}},\ \bibinfo {pages} {34--65} (\bibinfo {year} {2016})}\BibitemShut
  {NoStop}%
\bibitem [{\citenamefont {Jin}\ \emph {et~al.}(2021)\citenamefont {Jin},
  \citenamefont {Willsch}, \citenamefont {Willsch}, \citenamefont {Lagemann},
  \citenamefont {Michielsen},\ and\ \citenamefont
  {De~Raedt}}]{jin_willsch_willsch_lagemann_michielsen_deraedt_21}%
  \BibitemOpen
  \bibfield  {author} {\bibinfo {author} {\bibfnamefont {F.}~\bibnamefont
  {Jin}}, \bibinfo {author} {\bibfnamefont {D.}~\bibnamefont {Willsch}},
  \bibinfo {author} {\bibfnamefont {M.}~\bibnamefont {Willsch}}, \bibinfo
  {author} {\bibfnamefont {H.}~\bibnamefont {Lagemann}}, \bibinfo {author}
  {\bibfnamefont {K.}~\bibnamefont {Michielsen}},\ and\ \bibinfo {author}
  {\bibfnamefont {H.}~\bibnamefont {De~Raedt}},\ }\bibfield  {title} {\enquote
  {\bibinfo {title} {Random state technology},}\ }\href
  {https://doi.org/10.7566/jpsj.90.012001} {\bibfield  {journal} {\bibinfo
  {journal} {Journal of the Physical Society of Japan}\ }\textbf {\bibinfo
  {volume} {90}},\ \bibinfo {pages} {012001} (\bibinfo {year}
  {2021})}\BibitemShut {NoStop}%
\bibitem [{\citenamefont {Silver}\ and\ \citenamefont
  {R\"{o}der}(1994)}]{silver_roder_94}%
  \BibitemOpen
  \bibfield  {author} {\bibinfo {author} {\bibfnamefont {R.}~\bibnamefont
  {Silver}}\ and\ \bibinfo {author} {\bibfnamefont {H.}~\bibnamefont
  {R\"{o}der}},\ }\bibfield  {title} {\enquote {\bibinfo {title} {Densities of
  states of mega-dimensional {H}amiltonian matrices},}\ }\href
  {https://doi.org/10.1142/s0129183194000842} {\bibfield  {journal} {\bibinfo
  {journal} {International Journal of Modern Physics C}\ }\textbf {\bibinfo
  {volume} {05}},\ \bibinfo {pages} {735--753} (\bibinfo {year}
  {1994})}\BibitemShut {NoStop}%
\bibitem [{\citenamefont {Silver}\ \emph {et~al.}(1996)\citenamefont {Silver},
  \citenamefont {Roeder}, \citenamefont {Voter},\ and\ \citenamefont
  {Kress}}]{silver_roeder_voter_kress_96}%
  \BibitemOpen
  \bibfield  {author} {\bibinfo {author} {\bibfnamefont {R.}~\bibnamefont
  {Silver}}, \bibinfo {author} {\bibfnamefont {H.}~\bibnamefont {Roeder}},
  \bibinfo {author} {\bibfnamefont {A.}~\bibnamefont {Voter}},\ and\ \bibinfo
  {author} {\bibfnamefont {J.}~\bibnamefont {Kress}},\ }\bibfield  {title}
  {\enquote {\bibinfo {title} {Kernel polynomial approximations for densities
  of states and spectral functions},}\ }\href
  {https://doi.org/10.1006/jcph.1996.0048} {\bibfield  {journal} {\bibinfo
  {journal} {Journal of Computational Physics}\ }\textbf {\bibinfo {volume}
  {124}},\ \bibinfo {pages} {115--130} (\bibinfo {year} {1996})}\BibitemShut
  {NoStop}%
\bibitem [{\citenamefont {Wei{\ss}e}\ \emph {et~al.}(2006)\citenamefont
  {Wei{\ss}e}, \citenamefont {Wellein}, \citenamefont {Alvermann},\ and\
  \citenamefont {Fehske}}]{weisse_wellein_alvermann_fehske_06}%
  \BibitemOpen
  \bibfield  {author} {\bibinfo {author} {\bibfnamefont {A.}~\bibnamefont
  {Wei{\ss}e}}, \bibinfo {author} {\bibfnamefont {G.}~\bibnamefont {Wellein}},
  \bibinfo {author} {\bibfnamefont {A.}~\bibnamefont {Alvermann}},\ and\
  \bibinfo {author} {\bibfnamefont {H.}~\bibnamefont {Fehske}},\ }\bibfield
  {title} {\enquote {\bibinfo {title} {The kernel polynomial method},}\ }\href
  {https://doi.org/10.1103/revmodphys.78.275} {\bibfield  {journal} {\bibinfo
  {journal} {Reviews of Modern Physics}\ }\textbf {\bibinfo {volume} {78}},\
  \bibinfo {pages} {275--306} (\bibinfo {year} {2006})}\BibitemShut {NoStop}%
\bibitem [{\citenamefont {Ganeshan}, \citenamefont {Pixley},\ and\
  \citenamefont {Sarma}(2015)}]{ganeshan_pixley_dassarma_15}%
  \BibitemOpen
  \bibfield  {author} {\bibinfo {author} {\bibfnamefont {S.}~\bibnamefont
  {Ganeshan}}, \bibinfo {author} {\bibfnamefont {J.}~\bibnamefont {Pixley}},\
  and\ \bibinfo {author} {\bibfnamefont {S.~D.}\ \bibnamefont {Sarma}},\
  }\bibfield  {title} {\enquote {\bibinfo {title} {Nearest neighbor tight
  binding models with an exact mobility edge in one dimension},}\ }\href
  {https://doi.org/10.1103/physrevlett.114.146601} {\bibfield  {journal}
  {\bibinfo  {journal} {Physical Review Letters}\ }\textbf {\bibinfo {volume}
  {114}} (\bibinfo {year} {2015}),\ 10.1103/physrevlett.114.146601}\BibitemShut
  {NoStop}%
\bibitem [{\citenamefont {Garc{\'{\i}}a}, \citenamefont {Covaci},\ and\
  \citenamefont {Rappoport}(2015)}]{garca_covaci_rappoport_15}%
  \BibitemOpen
  \bibfield  {author} {\bibinfo {author} {\bibfnamefont {J.~H.}\ \bibnamefont
  {Garc{\'{\i}}a}}, \bibinfo {author} {\bibfnamefont {L.}~\bibnamefont
  {Covaci}},\ and\ \bibinfo {author} {\bibfnamefont {T.~G.}\ \bibnamefont
  {Rappoport}},\ }\bibfield  {title} {\enquote {\bibinfo {title} {Real-space
  calculation of the conductivity tensor for disordered topological matter},}\
  }\href {https://doi.org/10.1103/physrevlett.114.116602} {\bibfield  {journal}
  {\bibinfo  {journal} {Physical Review Letters}\ }\textbf {\bibinfo {volume}
  {114}} (\bibinfo {year} {2015}),\ 10.1103/physrevlett.114.116602}\BibitemShut
  {NoStop}%
\bibitem [{\citenamefont {Carr}\ \emph {et~al.}(2017)\citenamefont {Carr},
  \citenamefont {Massatt}, \citenamefont {Fang}, \citenamefont {Cazeaux},
  \citenamefont {Luskin},\ and\ \citenamefont
  {Kaxiras}}]{carr_massatt_fang_cazeaux_luskin_kaxiras_17}%
  \BibitemOpen
  \bibfield  {author} {\bibinfo {author} {\bibfnamefont {S.}~\bibnamefont
  {Carr}}, \bibinfo {author} {\bibfnamefont {D.}~\bibnamefont {Massatt}},
  \bibinfo {author} {\bibfnamefont {S.}~\bibnamefont {Fang}}, \bibinfo {author}
  {\bibfnamefont {P.}~\bibnamefont {Cazeaux}}, \bibinfo {author} {\bibfnamefont
  {M.}~\bibnamefont {Luskin}},\ and\ \bibinfo {author} {\bibfnamefont
  {E.}~\bibnamefont {Kaxiras}},\ }\bibfield  {title} {\enquote {\bibinfo
  {title} {Twistronics: Manipulating the electronic properties of
  two-dimensional layered structures through their twist angle},}\ }\href
  {https://doi.org/10.1103/physrevb.95.075420} {\bibfield  {journal} {\bibinfo
  {journal} {Physical Review B}\ }\textbf {\bibinfo {volume} {95}} (\bibinfo
  {year} {2017}),\ 10.1103/physrevb.95.075420}\BibitemShut {NoStop}%
\bibitem [{\citenamefont {Carvalho}\ \emph {et~al.}(2018)\citenamefont
  {Carvalho}, \citenamefont {Garc{\'{\i}}a-Mart{\'{\i}}nez}, \citenamefont
  {Lado},\ and\ \citenamefont
  {Fern{\'{a}}ndez-Rossier}}]{carvalho_garciamartinez_lado_fernandezrossier_18}%
  \BibitemOpen
  \bibfield  {author} {\bibinfo {author} {\bibfnamefont {D.}~\bibnamefont
  {Carvalho}}, \bibinfo {author} {\bibfnamefont {N.~A.}\ \bibnamefont
  {Garc{\'{\i}}a-Mart{\'{\i}}nez}}, \bibinfo {author} {\bibfnamefont {J.~L.}\
  \bibnamefont {Lado}},\ and\ \bibinfo {author} {\bibfnamefont
  {J.}~\bibnamefont {Fern{\'{a}}ndez-Rossier}},\ }\bibfield  {title} {\enquote
  {\bibinfo {title} {Real-space mapping of topological invariants using
  artificial neural networks},}\ }\href
  {https://doi.org/10.1103/physrevb.97.115453} {\bibfield  {journal} {\bibinfo
  {journal} {Physical Review B}\ }\textbf {\bibinfo {volume} {97}} (\bibinfo
  {year} {2018}),\ 10.1103/physrevb.97.115453}\BibitemShut {NoStop}%
\bibitem [{\citenamefont {Varjas}\ \emph {et~al.}(2020)\citenamefont {Varjas},
  \citenamefont {Fruchart}, \citenamefont {Akhmerov},\ and\ \citenamefont
  {Perez-Piskunow}}]{varjas_fruchart_akhmerov_perezpiskunow_20}%
  \BibitemOpen
  \bibfield  {author} {\bibinfo {author} {\bibfnamefont {D.}~\bibnamefont
  {Varjas}}, \bibinfo {author} {\bibfnamefont {M.}~\bibnamefont {Fruchart}},
  \bibinfo {author} {\bibfnamefont {A.~R.}\ \bibnamefont {Akhmerov}},\ and\
  \bibinfo {author} {\bibfnamefont {P.~M.}\ \bibnamefont {Perez-Piskunow}},\
  }\bibfield  {title} {\enquote {\bibinfo {title} {Computation of topological
  phase diagram of disordered
  ${\mathrm{pb}}_{1\ensuremath{-}x}{\mathrm{sn}}_{x}\mathrm{Te}$ using the
  kernel polynomial method},}\ }\href
  {https://doi.org/10.1103/physrevresearch.2.013229} {\bibfield  {journal}
  {\bibinfo  {journal} {Physical Review Research}\ }\textbf {\bibinfo {volume}
  {2}} (\bibinfo {year} {2020}),\ 10.1103/physrevresearch.2.013229}\BibitemShut
  {NoStop}%
\bibitem [{\citenamefont {Han}\ \emph {et~al.}(2017)\citenamefont {Han},
  \citenamefont {Malioutov}, \citenamefont {Avron},\ and\ \citenamefont
  {Shin}}]{han_malioutov_avron_shin_17}%
  \BibitemOpen
  \bibfield  {author} {\bibinfo {author} {\bibfnamefont {I.}~\bibnamefont
  {Han}}, \bibinfo {author} {\bibfnamefont {D.}~\bibnamefont {Malioutov}},
  \bibinfo {author} {\bibfnamefont {H.}~\bibnamefont {Avron}},\ and\ \bibinfo
  {author} {\bibfnamefont {J.}~\bibnamefont {Shin}},\ }\bibfield  {title}
  {\enquote {\bibinfo {title} {Approximating spectral sums of large-scale
  matrices using stochastic chebyshev approximations},}\ }\href
  {https://doi.org/10.1137/16m1078148} {\bibfield  {journal} {\bibinfo
  {journal} {{SIAM} Journal on Scientific Computing}\ }\textbf {\bibinfo
  {volume} {39}},\ \bibinfo {pages} {A1558--A1585} (\bibinfo {year}
  {2017})}\BibitemShut {NoStop}%
\bibitem [{\citenamefont {Dong}, \citenamefont {Benson},\ and\ \citenamefont
  {Bindel}(2019)}]{dong_benson_bindel_19}%
  \BibitemOpen
  \bibfield  {author} {\bibinfo {author} {\bibfnamefont {K.}~\bibnamefont
  {Dong}}, \bibinfo {author} {\bibfnamefont {A.~R.}\ \bibnamefont {Benson}},\
  and\ \bibinfo {author} {\bibfnamefont {D.}~\bibnamefont {Bindel}},\
  }\bibfield  {title} {\enquote {\bibinfo {title} {Network density of
  states},}\ }in\ \href {https://doi.org/10.1145/3292500.3330891} {\emph
  {\bibinfo {booktitle} {Proceedings of the 25th {ACM} {SIGKDD} International
  Conference on Knowledge Discovery {\&} Data Mining}}}\ (\bibinfo  {publisher}
  {{ACM}},\ \bibinfo {year} {2019})\BibitemShut {NoStop}%
\bibitem [{\citenamefont {Aichhorn}\ \emph {et~al.}(2003)\citenamefont
  {Aichhorn}, \citenamefont {Daghofer}, \citenamefont {Evertz},\ and\
  \citenamefont {von~der Linden}}]{aichhorn_daghofer_evertz_vondelinden_03}%
  \BibitemOpen
  \bibfield  {author} {\bibinfo {author} {\bibfnamefont {M.}~\bibnamefont
  {Aichhorn}}, \bibinfo {author} {\bibfnamefont {M.}~\bibnamefont {Daghofer}},
  \bibinfo {author} {\bibfnamefont {H.~G.}\ \bibnamefont {Evertz}},\ and\
  \bibinfo {author} {\bibfnamefont {W.}~\bibnamefont {von~der Linden}},\
  }\bibfield  {title} {\enquote {\bibinfo {title} {Low-temperature {L}anczos
  method for strongly correlated systems},}\ }\href
  {https://doi.org/10.1103/physrevb.67.161103} {\bibfield  {journal} {\bibinfo
  {journal} {Physical Review B}\ }\textbf {\bibinfo {volume} {67}} (\bibinfo
  {year} {2003}),\ 10.1103/physrevb.67.161103}\BibitemShut {NoStop}%
\bibitem [{\citenamefont {Ubaru}, \citenamefont {Chen},\ and\ \citenamefont
  {Saad}(2017)}]{ubaru_chen_saad_17}%
  \BibitemOpen
  \bibfield  {author} {\bibinfo {author} {\bibfnamefont {S.}~\bibnamefont
  {Ubaru}}, \bibinfo {author} {\bibfnamefont {J.}~\bibnamefont {Chen}},\ and\
  \bibinfo {author} {\bibfnamefont {Y.}~\bibnamefont {Saad}},\ }\bibfield
  {title} {\enquote {\bibinfo {title} {Fast estimation of
  $\operatorname{tr}(f(a))$ via stochastic lanczos quadrature},}\ }\href
  {https://doi.org/10.1137/16m1104974} {\bibfield  {journal} {\bibinfo
  {journal} {{SIAM} Journal on Matrix Analysis and Applications}\ }\textbf
  {\bibinfo {volume} {38}},\ \bibinfo {pages} {1075--1099} (\bibinfo {year}
  {2017})}\BibitemShut {NoStop}%
\bibitem [{\citenamefont {Granziol}, \citenamefont {Wan},\ and\ \citenamefont
  {Garipov}(2019)}]{granziol_wan_garipov_19}%
  \BibitemOpen
  \bibfield  {author} {\bibinfo {author} {\bibfnamefont {D.}~\bibnamefont
  {Granziol}}, \bibinfo {author} {\bibfnamefont {X.}~\bibnamefont {Wan}},\ and\
  \bibinfo {author} {\bibfnamefont {T.}~\bibnamefont {Garipov}},\ }\href@noop
  {} {\enquote {\bibinfo {title} {Deep curvature suite},}\ } (\bibinfo {year}
  {2019}),\ \Eprint {https://arxiv.org/abs/1912.09656} {arXiv:1912.09656
  [stat.ML]} \BibitemShut {NoStop}%
\bibitem [{\citenamefont {Schnack}, \citenamefont {Richter},\ and\
  \citenamefont {Steinigeweg}(2020)}]{schnack_richter_steinigeweg_20}%
  \BibitemOpen
  \bibfield  {author} {\bibinfo {author} {\bibfnamefont {J.}~\bibnamefont
  {Schnack}}, \bibinfo {author} {\bibfnamefont {J.}~\bibnamefont {Richter}},\
  and\ \bibinfo {author} {\bibfnamefont {R.}~\bibnamefont {Steinigeweg}},\
  }\bibfield  {title} {\enquote {\bibinfo {title} {Accuracy of the
  finite-temperature lanczos method compared to simple typicality-based
  estimates},}\ }\href {https://doi.org/10.1103/physrevresearch.2.013186}
  {\bibfield  {journal} {\bibinfo  {journal} {Physical Review Research}\
  }\textbf {\bibinfo {volume} {2}} (\bibinfo {year} {2020}),\
  10.1103/physrevresearch.2.013186}\BibitemShut {NoStop}%
\bibitem [{\citenamefont {Morita}\ and\ \citenamefont
  {Tohyama}(2020)}]{morita_tohyama_20}%
  \BibitemOpen
  \bibfield  {author} {\bibinfo {author} {\bibfnamefont {K.}~\bibnamefont
  {Morita}}\ and\ \bibinfo {author} {\bibfnamefont {T.}~\bibnamefont
  {Tohyama}},\ }\bibfield  {title} {\enquote {\bibinfo {title}
  {Finite-temperature properties of the {K}itaev-{H}eisenberg models on kagome
  and triangular lattices studied by improved finite-temperature {L}anczos
  methods},}\ }\href {https://doi.org/10.1103/physrevresearch.2.013205}
  {\bibfield  {journal} {\bibinfo  {journal} {Physical Review Research}\
  }\textbf {\bibinfo {volume} {2}} (\bibinfo {year} {2020}),\
  10.1103/physrevresearch.2.013205}\BibitemShut {NoStop}%
\bibitem [{\citenamefont {Chen}, \citenamefont {Trogdon},\ and\ \citenamefont
  {Ubaru}(2022)}]{chen_trogdon_ubaru_22}%
  \BibitemOpen
  \bibfield  {author} {\bibinfo {author} {\bibfnamefont {T.}~\bibnamefont
  {Chen}}, \bibinfo {author} {\bibfnamefont {T.}~\bibnamefont {Trogdon}},\ and\
  \bibinfo {author} {\bibfnamefont {S.}~\bibnamefont {Ubaru}},\ }\href@noop {}
  {\enquote {\bibinfo {title} {Randomized matrix-free quadrature for spectrum
  and spectral sum approximation},}\ } (\bibinfo {year} {2022}),\ \Eprint
  {https://arxiv.org/abs/2204.01941} {arXiv:2204.01941 [math.NA]} \BibitemShut
  {NoStop}%
\bibitem [{\citenamefont {Goldstein}\ \emph {et~al.}(2010)\citenamefont
  {Goldstein}, \citenamefont {Lebowitz}, \citenamefont {Mastrodonato},
  \citenamefont {Tumulka},\ and\ \citenamefont
  {Zangh{\`{\i}}}}]{goldstein_lebowitz_mastrodonato_tumulka_zanghi_10}%
  \BibitemOpen
  \bibfield  {author} {\bibinfo {author} {\bibfnamefont {S.}~\bibnamefont
  {Goldstein}}, \bibinfo {author} {\bibfnamefont {J.~L.}\ \bibnamefont
  {Lebowitz}}, \bibinfo {author} {\bibfnamefont {C.}~\bibnamefont
  {Mastrodonato}}, \bibinfo {author} {\bibfnamefont {R.}~\bibnamefont
  {Tumulka}},\ and\ \bibinfo {author} {\bibfnamefont {N.}~\bibnamefont
  {Zangh{\`{\i}}}},\ }\bibfield  {title} {\enquote {\bibinfo {title} {Normal
  typicality and von neumann's quantum ergodic theorem},}\ }\href
  {https://doi.org/10.1098/rspa.2009.0635} {\bibfield  {journal} {\bibinfo
  {journal} {Proceedings of the Royal Society A: Mathematical, Physical and
  Engineering Sciences}\ }\textbf {\bibinfo {volume} {466}},\ \bibinfo {pages}
  {3203--3224} (\bibinfo {year} {2010})}\BibitemShut {NoStop}%
\bibitem [{\citenamefont {Alben}\ \emph {et~al.}(1975)\citenamefont {Alben},
  \citenamefont {Blume}, \citenamefont {Krakauer},\ and\ \citenamefont
  {Schwartz}}]{alben_blume_krakauer_schwartz_75}%
  \BibitemOpen
  \bibfield  {author} {\bibinfo {author} {\bibfnamefont {R.}~\bibnamefont
  {Alben}}, \bibinfo {author} {\bibfnamefont {M.}~\bibnamefont {Blume}},
  \bibinfo {author} {\bibfnamefont {H.}~\bibnamefont {Krakauer}},\ and\
  \bibinfo {author} {\bibfnamefont {L.}~\bibnamefont {Schwartz}},\ }\bibfield
  {title} {\enquote {\bibinfo {title} {Exact results for a three-dimensional
  alloy with site diagonal disorder: comparison with the coherent potential
  approximation},}\ }\href {https://doi.org/10.1103/physrevb.12.4090}
  {\bibfield  {journal} {\bibinfo  {journal} {Physical Review B}\ }\textbf
  {\bibinfo {volume} {12}},\ \bibinfo {pages} {4090--4094} (\bibinfo {year}
  {1975})}\BibitemShut {NoStop}%
\bibitem [{\citenamefont {Schnack}\ \emph {et~al.}(2020)\citenamefont
  {Schnack}, \citenamefont {Richter}, \citenamefont {Heitmann}, \citenamefont
  {Richter},\ and\ \citenamefont
  {Steinigeweg}}]{schnack_richter_heitmann_richter_steinigeweg_20}%
  \BibitemOpen
  \bibfield  {author} {\bibinfo {author} {\bibfnamefont {J.}~\bibnamefont
  {Schnack}}, \bibinfo {author} {\bibfnamefont {J.}~\bibnamefont {Richter}},
  \bibinfo {author} {\bibfnamefont {T.}~\bibnamefont {Heitmann}}, \bibinfo
  {author} {\bibfnamefont {J.}~\bibnamefont {Richter}},\ and\ \bibinfo {author}
  {\bibfnamefont {R.}~\bibnamefont {Steinigeweg}},\ }\bibfield  {title}
  {\enquote {\bibinfo {title} {Finite-size scaling of typicality-based
  estimates},}\ }\href {https://doi.org/10.1515/zna-2020-0031} {\bibfield
  {journal} {\bibinfo  {journal} {Zeitschrift f\"{u}r Naturforschung A}\
  }\textbf {\bibinfo {volume} {75}},\ \bibinfo {pages} {465--473} (\bibinfo
  {year} {2020})}\BibitemShut {NoStop}%
\bibitem [{\citenamefont {Girard}(1987)}]{girard_87}%
  \BibitemOpen
  \bibfield  {author} {\bibinfo {author} {\bibfnamefont {D.}~\bibnamefont
  {Girard}},\ }\href@noop {} {\enquote {\bibinfo {title} {Un algorithme simple
  et rapide pour la validation crois{\'e}e g{\'e}n{\'e}ralis{\'e}e sur des
  probl{\`e}mes de grande taille},}\ } (\bibinfo {year} {1987})\BibitemShut
  {NoStop}%
\bibitem [{\citenamefont {Girard}(1989)}]{girard_89}%
  \BibitemOpen
  \bibfield  {author} {\bibinfo {author} {\bibfnamefont {A.}~\bibnamefont
  {Girard}},\ }\bibfield  {title} {\enquote {\bibinfo {title} {A fast
  ?monte-carlo cross-validation? procedure for large least squares problems
  with noisy data},}\ }\href {https://doi.org/10.1007/bf01395775} {\bibfield
  {journal} {\bibinfo  {journal} {Numerische Mathematik}\ }\textbf {\bibinfo
  {volume} {56}},\ \bibinfo {pages} {1--23} (\bibinfo {year}
  {1989})}\BibitemShut {NoStop}%
\bibitem [{\citenamefont {Hutchinson}(1989)}]{hutchinson_89}%
  \BibitemOpen
  \bibfield  {author} {\bibinfo {author} {\bibfnamefont {M.}~\bibnamefont
  {Hutchinson}},\ }\bibfield  {title} {\enquote {\bibinfo {title} {A stochastic
  estimator of the trace of the influence matrix for laplacian smoothing
  splines},}\ }\href {https://doi.org/10.1080/03610918908812806} {\bibfield
  {journal} {\bibinfo  {journal} {Communications in Statistics - Simulation and
  Computation}\ }\textbf {\bibinfo {volume} {18}},\ \bibinfo {pages}
  {1059--1076} (\bibinfo {year} {1989})}\BibitemShut {NoStop}%
\bibitem [{\citenamefont {Avron}\ and\ \citenamefont
  {Toledo}(2011)}]{avron_toledo_11}%
  \BibitemOpen
  \bibfield  {author} {\bibinfo {author} {\bibfnamefont {H.}~\bibnamefont
  {Avron}}\ and\ \bibinfo {author} {\bibfnamefont {S.}~\bibnamefont {Toledo}},\
  }\bibfield  {title} {\enquote {\bibinfo {title} {Randomized algorithms for
  estimating the trace of an implicit symmetric positive semi-definite
  matrix},}\ }\href {https://doi.org/10.1145/1944345.1944349} {\bibfield
  {journal} {\bibinfo  {journal} {Journal of the {ACM}}\ }\textbf {\bibinfo
  {volume} {58}},\ \bibinfo {pages} {1--34} (\bibinfo {year}
  {2011})}\BibitemShut {NoStop}%
\bibitem [{\citenamefont {Roosta-Khorasani}\ and\ \citenamefont
  {Ascher}(2014)}]{roostakhorasani_ascher_14}%
  \BibitemOpen
  \bibfield  {author} {\bibinfo {author} {\bibfnamefont {F.}~\bibnamefont
  {Roosta-Khorasani}}\ and\ \bibinfo {author} {\bibfnamefont {U.}~\bibnamefont
  {Ascher}},\ }\bibfield  {title} {\enquote {\bibinfo {title} {Improved bounds
  on sample size for implicit matrix trace estimators},}\ }\href
  {https://doi.org/10.1007/s10208-014-9220-1} {\bibfield  {journal} {\bibinfo
  {journal} {Foundations of Computational Mathematics}\ }\textbf {\bibinfo
  {volume} {15}},\ \bibinfo {pages} {1187--1212} (\bibinfo {year}
  {2014})}\BibitemShut {NoStop}%
\bibitem [{\citenamefont {Cortinovis}\ and\ \citenamefont
  {Kressner}(2021)}]{cortinovis_kressner_21}%
  \BibitemOpen
  \bibfield  {author} {\bibinfo {author} {\bibfnamefont {A.}~\bibnamefont
  {Cortinovis}}\ and\ \bibinfo {author} {\bibfnamefont {D.}~\bibnamefont
  {Kressner}},\ }\bibfield  {title} {\enquote {\bibinfo {title} {On randomized
  trace estimates for indefinite matrices with an application to
  determinants},}\ }\href {https://doi.org/10.1007/s10208-021-09525-9}
  {\bibfield  {journal} {\bibinfo  {journal} {Foundations of Computational
  Mathematics}\ } (\bibinfo {year} {2021}),\
  10.1007/s10208-021-09525-9}\BibitemShut {NoStop}%
\bibitem [{Note1()}]{Note1}%
  \BibitemOpen
  \bibinfo {note} {If we can compute the moments of the DOS $\rho (E)$, then we
  can apply the KPM to the DOS directly.}\BibitemShut {Stop}%
\bibitem [{\citenamefont {Braverman}, \citenamefont {Krishnan},\ and\
  \citenamefont {Musco}(2022)}]{braverman_krishnan_musco_22}%
  \BibitemOpen
  \bibfield  {author} {\bibinfo {author} {\bibfnamefont {V.}~\bibnamefont
  {Braverman}}, \bibinfo {author} {\bibfnamefont {A.}~\bibnamefont
  {Krishnan}},\ and\ \bibinfo {author} {\bibfnamefont {C.}~\bibnamefont
  {Musco}},\ }\bibfield  {title} {\enquote {\bibinfo {title} {Sublinear time
  spectral density estimation},}\ }in\ \href
  {https://doi.org/10.1145/3519935.3520009} {\emph {\bibinfo {booktitle}
  {Proceedings of the 54th Annual {ACM} {SIGACT} Symposium on Theory of
  Computing}}}\ (\bibinfo  {publisher} {{ACM}},\ \bibinfo {year} {2022})\
  \bibinfo {note} {arXiv cs.DS 2104.03461},\ \Eprint
  {https://arxiv.org/abs/2104.03461} {arXiv:2104.03461 [cs.DS]} \BibitemShut
  {NoStop}%
\bibitem [{\citenamefont {Schl\"{u}ter}\ \emph {et~al.}(2021)\citenamefont
  {Schl\"{u}ter}, \citenamefont {Gayk}, \citenamefont {Schmidt}, \citenamefont
  {Honecker},\ and\ \citenamefont
  {Schnack}}]{schluter_gayk_schmidt_honecker_schnack_21}%
  \BibitemOpen
  \bibfield  {author} {\bibinfo {author} {\bibfnamefont {H.}~\bibnamefont
  {Schl\"{u}ter}}, \bibinfo {author} {\bibfnamefont {F.}~\bibnamefont {Gayk}},
  \bibinfo {author} {\bibfnamefont {H.-J.}\ \bibnamefont {Schmidt}}, \bibinfo
  {author} {\bibfnamefont {A.}~\bibnamefont {Honecker}},\ and\ \bibinfo
  {author} {\bibfnamefont {J.}~\bibnamefont {Schnack}},\ }\bibfield  {title}
  {\enquote {\bibinfo {title} {Accuracy of the typicality approach using
  chebyshev polynomials},}\ }\href {https://doi.org/10.1515/zna-2021-0116}
  {\bibfield  {journal} {\bibinfo  {journal} {Zeitschrift f\"{u}r
  Naturforschung A}\ }\textbf {\bibinfo {volume} {76}},\ \bibinfo {pages}
  {823--834} (\bibinfo {year} {2021})}\BibitemShut {NoStop}%
\bibitem [{\citenamefont {Saad}(1983)}]{saad_83}%
  \BibitemOpen
  \bibfield  {author} {\bibinfo {author} {\bibfnamefont {Y.}~\bibnamefont
  {Saad}},\ }\bibfield  {title} {\enquote {\bibinfo {title} {Iterative solution
  of indefinite symmetric linear systems by methods using orthogonal
  polynomials over two disjoint intervals},}\ }\href
  {https://doi.org/10.1137/0720052} {\bibfield  {journal} {\bibinfo  {journal}
  {{SIAM} Journal on Numerical Analysis}\ }\textbf {\bibinfo {volume} {20}},\
  \bibinfo {pages} {784--811} (\bibinfo {year} {1983})}\BibitemShut {NoStop}%
\bibitem [{\citenamefont {Geronimo}\ and\ \citenamefont
  {Assche}(1988)}]{geronimo_vanassche_88}%
  \BibitemOpen
  \bibfield  {author} {\bibinfo {author} {\bibfnamefont {J.~S.}\ \bibnamefont
  {Geronimo}}\ and\ \bibinfo {author} {\bibfnamefont {W.~V.}\ \bibnamefont
  {Assche}},\ }\bibfield  {title} {\enquote {\bibinfo {title} {Orthogonal
  polynomials on several intervals via a polynomial mapping},}\ }\href
  {https://doi.org/10.1090/s0002-9947-1988-0951620-6} {\bibfield  {journal}
  {\bibinfo  {journal} {Transactions of the American Mathematical Society}\
  }\textbf {\bibinfo {volume} {308}},\ \bibinfo {pages} {559--581} (\bibinfo
  {year} {1988})}\BibitemShut {NoStop}%
\bibitem [{\citenamefont {Golub}\ and\ \citenamefont
  {Meurant}(2009)}]{golub_meurant_09}%
  \BibitemOpen
  \bibfield  {author} {\bibinfo {author} {\bibfnamefont {G.~H.}\ \bibnamefont
  {Golub}}\ and\ \bibinfo {author} {\bibfnamefont {G.}~\bibnamefont
  {Meurant}},\ }\href@noop {} {\emph {\bibinfo {title} {Matrices, moments and
  quadrature with applications}}},\ \bibinfo {series} {Princeton series in
  applied mathematics}, Vol.~\bibinfo {volume} {30}\ (\bibinfo  {publisher}
  {Princeton University Press},\ \bibinfo {year} {2009})\BibitemShut {NoStop}%
\bibitem [{\citenamefont {Kuczy{\'{n}}ski}\ and\ \citenamefont
  {Wo{\'{z}}niakowski}(1992)}]{kuczyski_wozniakowski_92}%
  \BibitemOpen
  \bibfield  {author} {\bibinfo {author} {\bibfnamefont {J.}~\bibnamefont
  {Kuczy{\'{n}}ski}}\ and\ \bibinfo {author} {\bibfnamefont {H.}~\bibnamefont
  {Wo{\'{z}}niakowski}},\ }\bibfield  {title} {\enquote {\bibinfo {title}
  {Estimating the largest eigenvalue by the power and {L}anczos algorithms with
  a random start},}\ }\href {https://doi.org/10.1137/0613066} {\bibfield
  {journal} {\bibinfo  {journal} {{SIAM} Journal on Matrix Analysis and
  Applications}\ }\textbf {\bibinfo {volume} {13}},\ \bibinfo {pages}
  {1094--1122} (\bibinfo {year} {1992})}\BibitemShut {NoStop}%
\bibitem [{\citenamefont {Zhou}\ and\ \citenamefont {Li}(2011)}]{zhou_li_11}%
  \BibitemOpen
  \bibfield  {author} {\bibinfo {author} {\bibfnamefont {Y.}~\bibnamefont
  {Zhou}}\ and\ \bibinfo {author} {\bibfnamefont {R.-C.}\ \bibnamefont {Li}},\
  }\bibfield  {title} {\enquote {\bibinfo {title} {Bounding the spectrum of
  large hermitian matrices},}\ }\href
  {https://doi.org/10.1016/j.laa.2010.06.034} {\bibfield  {journal} {\bibinfo
  {journal} {Linear Algebra and its Applications}\ }\textbf {\bibinfo {volume}
  {435}},\ \bibinfo {pages} {480--493} (\bibinfo {year} {2011})}\BibitemShut
  {NoStop}%
\bibitem [{\citenamefont {Martinsson}\ and\ \citenamefont
  {Tropp}(2020)}]{martinsson_tropp_20}%
  \BibitemOpen
  \bibfield  {author} {\bibinfo {author} {\bibfnamefont {P.-G.}\ \bibnamefont
  {Martinsson}}\ and\ \bibinfo {author} {\bibfnamefont {J.~A.}\ \bibnamefont
  {Tropp}},\ }\bibfield  {title} {\enquote {\bibinfo {title} {Randomized
  numerical linear algebra: Foundations and algorithms},}\ }\href
  {https://doi.org/10.1017/s0962492920000021} {\bibfield  {journal} {\bibinfo
  {journal} {Acta Numerica}\ }\textbf {\bibinfo {volume} {29}},\ \bibinfo
  {pages} {403--572} (\bibinfo {year} {2020})}\BibitemShut {NoStop}%
\bibitem [{\citenamefont {Paige}(1970)}]{paige_70}%
  \BibitemOpen
  \bibfield  {author} {\bibinfo {author} {\bibfnamefont {C.~C.}\ \bibnamefont
  {Paige}},\ }\bibfield  {title} {\enquote {\bibinfo {title} {Practical use of
  the symmetric lanczos process with re-orthogonalization},}\ }\href
  {https://doi.org/10.1007/bf01936866} {\bibfield  {journal} {\bibinfo
  {journal} {{BIT}}\ }\textbf {\bibinfo {volume} {10}},\ \bibinfo {pages}
  {183--195} (\bibinfo {year} {1970})}\BibitemShut {NoStop}%
\bibitem [{\citenamefont {Paige}(1976)}]{paige_76}%
  \BibitemOpen
  \bibfield  {author} {\bibinfo {author} {\bibfnamefont {C.~C.}\ \bibnamefont
  {Paige}},\ }\bibfield  {title} {\enquote {\bibinfo {title} {{Error Analysis
  of the Lanczos Algorithm for Tridiagonalizing a Symmetric Matrix}},}\ }\href
  {https://doi.org/10.1093/imamat/18.3.341} {\bibfield  {journal} {\bibinfo
  {journal} {IMA Journal of Applied Mathematics}\ }\textbf {\bibinfo {volume}
  {18}},\ \bibinfo {pages} {341--349} (\bibinfo {year} {1976})}\BibitemShut
  {NoStop}%
\bibitem [{\citenamefont {Paige}(1980)}]{paige_80}%
  \BibitemOpen
  \bibfield  {author} {\bibinfo {author} {\bibfnamefont {C.~C.}\ \bibnamefont
  {Paige}},\ }\bibfield  {title} {\enquote {\bibinfo {title} {Accuracy and
  effectiveness of the {L}anczos algorithm for the symmetric eigenproblem},}\
  }\href {https://doi.org/10.1016/0024-3795(80)90167-6} {\bibfield  {journal}
  {\bibinfo  {journal} {Linear Algebra and its Applications}\ }\textbf
  {\bibinfo {volume} {34}},\ \bibinfo {pages} {235 -- 258} (\bibinfo {year}
  {1980})}\BibitemShut {NoStop}%
\bibitem [{\citenamefont {Greenbaum}(1989)}]{greenbaum_89}%
  \BibitemOpen
  \bibfield  {author} {\bibinfo {author} {\bibfnamefont {A.}~\bibnamefont
  {Greenbaum}},\ }\bibfield  {title} {\enquote {\bibinfo {title} {Behavior of
  slightly perturbed {L}anczos and conjugate-gradient recurrences},}\ }\href
  {https://doi.org/10.1016/0024-3795(89)90285-1} {\bibfield  {journal}
  {\bibinfo  {journal} {Linear Algebra and its Applications}\ }\textbf
  {\bibinfo {volume} {113}},\ \bibinfo {pages} {7 -- 63} (\bibinfo {year}
  {1989})}\BibitemShut {NoStop}%
\bibitem [{\citenamefont {Strakos}\ and\ \citenamefont
  {Greenbaum}(1992)}]{strakos_greenbaum_92}%
  \BibitemOpen
  \bibfield  {author} {\bibinfo {author} {\bibfnamefont {Z.}~\bibnamefont
  {Strakos}}\ and\ \bibinfo {author} {\bibfnamefont {A.}~\bibnamefont
  {Greenbaum}},\ }\bibfield  {title} {\enquote {\bibinfo {title} {Open
  questions in the convergence analysis of the {L}anczos process for the real
  symmetric eigenvalue problem},}\ \ }(\bibinfo  {publisher} {University of
  Minnesota},\ \bibinfo {year} {1992})\ \Eprint
  {https://arxiv.org/abs/https://conservancy.umn.edu/handle/11299/1838}
  {https://conservancy.umn.edu/handle/11299/1838} \BibitemShut {NoStop}%
\bibitem [{\citenamefont {Druskin}\ and\ \citenamefont
  {Knizhnerman}(1991)}]{druskin_knizhnerman_91}%
  \BibitemOpen
  \bibfield  {author} {\bibinfo {author} {\bibfnamefont {V.~L.}\ \bibnamefont
  {Druskin}}\ and\ \bibinfo {author} {\bibfnamefont {L.~A.}\ \bibnamefont
  {Knizhnerman}},\ }\bibfield  {title} {\enquote {\bibinfo {title} {Error
  bounds in the simple lanczos procedure for computing functions of symmetric
  matrices and eigenvalues},}\ }\href@noop {} {\bibfield  {journal} {\bibinfo
  {journal} {Comput. Math. Math. Phys.}\ }\textbf {\bibinfo {volume} {31}},\
  \bibinfo {pages} {20–30} (\bibinfo {year} {1991})}\BibitemShut {NoStop}%
\bibitem [{\citenamefont {Long}\ \emph {et~al.}(2003)\citenamefont {Long},
  \citenamefont {Prelov{\v{s}}ek}, \citenamefont {Shawish}, \citenamefont
  {Karadamoglou},\ and\ \citenamefont
  {Zotos}}]{long_prelovsek_sawish_karadamoglou_03}%
  \BibitemOpen
  \bibfield  {author} {\bibinfo {author} {\bibfnamefont {M.~W.}\ \bibnamefont
  {Long}}, \bibinfo {author} {\bibfnamefont {P.}~\bibnamefont
  {Prelov{\v{s}}ek}}, \bibinfo {author} {\bibfnamefont {S.~E.}\ \bibnamefont
  {Shawish}}, \bibinfo {author} {\bibfnamefont {J.}~\bibnamefont
  {Karadamoglou}},\ and\ \bibinfo {author} {\bibfnamefont {X.}~\bibnamefont
  {Zotos}},\ }\bibfield  {title} {\enquote {\bibinfo {title}
  {Finite-temperature dynamical correlations using the microcanonical ensemble
  and the lanczos algorithm},}\ }\href
  {https://doi.org/10.1103/physrevb.68.235106} {\bibfield  {journal} {\bibinfo
  {journal} {Physical Review B}\ }\textbf {\bibinfo {volume} {68}} (\bibinfo
  {year} {2003}),\ 10.1103/physrevb.68.235106}\BibitemShut {NoStop}%
\bibitem [{\citenamefont {Chen}, \citenamefont {Trogdon},\ and\ \citenamefont
  {Ubaru}(2021)}]{chen_trogdon_ubaru_21}%
  \BibitemOpen
  \bibfield  {author} {\bibinfo {author} {\bibfnamefont {T.}~\bibnamefont
  {Chen}}, \bibinfo {author} {\bibfnamefont {T.}~\bibnamefont {Trogdon}},\ and\
  \bibinfo {author} {\bibfnamefont {S.}~\bibnamefont {Ubaru}},\ }\bibfield
  {title} {\enquote {\bibinfo {title} {Analysis of stochastic lanczos
  quadrature for spectrum approximation},}\ }in\ \href
  {http://proceedings.mlr.press/v139/chen21s.html} {\emph {\bibinfo {booktitle}
  {Proceedings of the 38th International Conference on Machine Learning}}},\
  \bibinfo {series} {Proceedings of Machine Learning Research}, Vol.\ \bibinfo
  {volume} {139}\ (\bibinfo  {publisher} {PMLR},\ \bibinfo {year} {2021})\ pp.\
  \bibinfo {pages} {1728--1739},\ \Eprint {https://arxiv.org/abs/2105.06595}
  {arXiv:2105.06595 [cs.DS]} \BibitemShut {NoStop}%
\bibitem [{\citenamefont {Knizhnerman}(1996)}]{knizhnerman_96}%
  \BibitemOpen
  \bibfield  {author} {\bibinfo {author} {\bibfnamefont {L.~A.}\ \bibnamefont
  {Knizhnerman}},\ }\bibfield  {title} {\enquote {\bibinfo {title} {The simple
  {L}anczos procedure: Estimates of the error of the {G}auss quadrature formula
  and their applications},}\ }\href@noop {} {\bibfield  {journal} {\bibinfo
  {journal} {Comput. Math. Math. Phys.}\ }\textbf {\bibinfo {volume} {36}},\
  \bibinfo {pages} {1481–1492} (\bibinfo {year} {1996})}\BibitemShut
  {NoStop}%
\bibitem [{\citenamefont {Karabach}\ \emph {et~al.}(1997)\citenamefont
  {Karabach}, \citenamefont {M\"{u}ller}, \citenamefont {Gould},\ and\
  \citenamefont {Tobochnik}}]{karabach_97}%
  \BibitemOpen
  \bibfield  {author} {\bibinfo {author} {\bibfnamefont {M.}~\bibnamefont
  {Karabach}}, \bibinfo {author} {\bibfnamefont {G.}~\bibnamefont
  {M\"{u}ller}}, \bibinfo {author} {\bibfnamefont {H.}~\bibnamefont {Gould}},\
  and\ \bibinfo {author} {\bibfnamefont {J.}~\bibnamefont {Tobochnik}},\
  }\bibfield  {title} {\enquote {\bibinfo {title} {Introduction to the {B}ethe
  ansatz i},}\ }\href {https://doi.org/10.1063/1.4822511} {\bibfield  {journal}
  {\bibinfo  {journal} {Computers in Physics}\ }\textbf {\bibinfo {volume}
  {11}},\ \bibinfo {pages} {36} (\bibinfo {year} {1997})}\BibitemShut {NoStop}%
\bibitem [{\citenamefont {Groth}\ \emph {et~al.}(2014)\citenamefont {Groth},
  \citenamefont {Wimmer}, \citenamefont {Akhmerov},\ and\ \citenamefont
  {Waintal}}]{groth_wimmer_akhmerov_waintal_14}%
  \BibitemOpen
  \bibfield  {author} {\bibinfo {author} {\bibfnamefont {C.~W.}\ \bibnamefont
  {Groth}}, \bibinfo {author} {\bibfnamefont {M.}~\bibnamefont {Wimmer}},
  \bibinfo {author} {\bibfnamefont {A.~R.}\ \bibnamefont {Akhmerov}},\ and\
  \bibinfo {author} {\bibfnamefont {X.}~\bibnamefont {Waintal}},\ }\bibfield
  {title} {\enquote {\bibinfo {title} {Kwant: a software package for quantum
  transport},}\ }\href {https://doi.org/10.1088/1367-2630/16/6/063065}
  {\bibfield  {journal} {\bibinfo  {journal} {New Journal of Physics}\ }\textbf
  {\bibinfo {volume} {16}},\ \bibinfo {pages} {063065} (\bibinfo {year}
  {2014})}\BibitemShut {NoStop}%
\bibitem [{\citenamefont {Kronik}\ \emph {et~al.}(2006)\citenamefont {Kronik},
  \citenamefont {Makmal}, \citenamefont {Tiago}, \citenamefont {Alemany},
  \citenamefont {Jain}, \citenamefont {Huang}, \citenamefont {Saad},\ and\
  \citenamefont
  {Chelikowsky}}]{kronik_makmal_tiago_alemany_jain_huang_saad_chelikowsky_06}%
  \BibitemOpen
  \bibfield  {author} {\bibinfo {author} {\bibfnamefont {L.}~\bibnamefont
  {Kronik}}, \bibinfo {author} {\bibfnamefont {A.}~\bibnamefont {Makmal}},
  \bibinfo {author} {\bibfnamefont {M.~L.}\ \bibnamefont {Tiago}}, \bibinfo
  {author} {\bibfnamefont {M.~M.~G.}\ \bibnamefont {Alemany}}, \bibinfo
  {author} {\bibfnamefont {M.}~\bibnamefont {Jain}}, \bibinfo {author}
  {\bibfnamefont {X.}~\bibnamefont {Huang}}, \bibinfo {author} {\bibfnamefont
  {Y.}~\bibnamefont {Saad}},\ and\ \bibinfo {author} {\bibfnamefont {J.~R.}\
  \bibnamefont {Chelikowsky}},\ }\bibfield  {title} {\enquote {\bibinfo {title}
  {{PARSEC} {\textendash} the pseudopotential algorithm for real-space
  electronic structure calculations: recent advances and novel applications to
  nano-structures},}\ }\href {https://doi.org/10.1002/pssb.200541463}
  {\bibfield  {journal} {\bibinfo  {journal} {physica status solidi (b)}\
  }\textbf {\bibinfo {volume} {243}},\ \bibinfo {pages} {1063--1079} (\bibinfo
  {year} {2006})}\BibitemShut {NoStop}%
\bibitem [{\citenamefont {Davis}\ and\ \citenamefont {Hu}(2011)}]{davis_hu_11}%
  \BibitemOpen
  \bibfield  {author} {\bibinfo {author} {\bibfnamefont {T.~A.}\ \bibnamefont
  {Davis}}\ and\ \bibinfo {author} {\bibfnamefont {Y.}~\bibnamefont {Hu}},\
  }\bibfield  {title} {\enquote {\bibinfo {title} {The university of florida
  sparse matrix collection},}\ }\href {https://doi.org/10.1145/2049662.2049663}
  {\bibfield  {journal} {\bibinfo  {journal} {{ACM} Transactions on
  Mathematical Software}\ }\textbf {\bibinfo {volume} {38}},\ \bibinfo {pages}
  {1--25} (\bibinfo {year} {2011})}\BibitemShut {NoStop}%
\bibitem [{\citenamefont {Li}\ \emph {et~al.}(2019)\citenamefont {Li},
  \citenamefont {Xi}, \citenamefont {Erlandson},\ and\ \citenamefont
  {Saad}}]{li_xi_erlandson_saad_19}%
  \BibitemOpen
  \bibfield  {author} {\bibinfo {author} {\bibfnamefont {R.}~\bibnamefont
  {Li}}, \bibinfo {author} {\bibfnamefont {Y.}~\bibnamefont {Xi}}, \bibinfo
  {author} {\bibfnamefont {L.}~\bibnamefont {Erlandson}},\ and\ \bibinfo
  {author} {\bibfnamefont {Y.}~\bibnamefont {Saad}},\ }\bibfield  {title}
  {\enquote {\bibinfo {title} {The eigenvalues slicing library ({EVSL}):
  Algorithms, implementation, and software},}\ }\href
  {https://doi.org/10.1137/18m1170935} {\bibfield  {journal} {\bibinfo
  {journal} {{SIAM} Journal on Scientific Computing}\ }\textbf {\bibinfo
  {volume} {41}},\ \bibinfo {pages} {C393--C415} (\bibinfo {year}
  {2019})}\BibitemShut {NoStop}%
\bibitem [{Note2()}]{Note2}%
  \BibitemOpen
  \bibinfo {note} {If the number of outlying eigenvalues is known, it is
  possible to deflate them.\cite
  {weisse_wellein_alvermann_fehske_06,morita_tohyama_20} However, this requires
  additional computation including the storage of the eigenvectors (which may
  be intractable).}\BibitemShut {Stop}%
\end{thebibliography}%

\end{document}